\documentclass[a4paper,10pt, fleqn]{article}

\usepackage{etex}
\usepackage[T1]{fontenc}
\usepackage[utf8]{inputenc}

\usepackage[a4paper,tmargin=3truecm,bmargin=3truecm,hmargin=1.6truecm]{geometry}

\usepackage{xcolor}

\usepackage[square,numbers,sort,compress,semicolon,merge]{natbib}
\usepackage[french, english]{babel}

\usepackage[bitstream-charter]{mathdesign}

\usepackage{array,booktabs}
\usepackage{enumitem}

\usepackage[colorlinks=true, pdfstartview=FitV, linkcolor=blue, citecolor=blue, urlcolor=blue]{hyperref}
\usepackage{notoccite}
\usepackage{amsmath,mathtools,bm}
\usepackage{amsthm}
\usepackage{dsfont}

\usepackage{cancel}
\usepackage{tikz}
\usepackage{tikz-cd}
\usetikzlibrary{cd}

\usepackage[shadow,backgroundcolor=red!20]{todonotes}

\newcounter{mnotecount}[section]
\renewcommand{\themnotecount}{\thesection.\arabic{mnotecount}}
\newcommand{\mnote}[1]%
{\protect{\stepcounter{mnotecount}}${}^{\text{\footnotesize$\bullet$\themnotecount}}$%
\marginpar{\raggedright\tiny$\!\!\!\!\!\!\,\bullet$\themnotecount: #1}}

\numberwithin{equation}{section}

\newcommand{\defeq}{\vcentcolon=}
\newcommand{\rdefeq}{=\vcentcolon}
\newcommand{\connexionShift}{{\boldsymbol\alpha}}
\newcommand{\hconnexionShift}{\widehat{\boldsymbol\alpha}}

\newcommand{\grast}{\bullet} 
\newcommand{\omi}[1]{\buildrel { \buildrel{#1}\over{\vee} } \over .} 

\newcommand\dd{\text{\textup{d}}} 
\newcommand\ds{\text{\textup{s}}} 

\newcommand\exter{{\textstyle\bigwedge}} 
\newcommand\locto{\rightsquigarrow} 

\newcommand{\algzero}{\mathsf{0}}

\newcommand\calH{\mathcal{H}}
\newcommand\calL{\mathcal{L}}
\newcommand\calM{\mathcal{M}}
\newcommand\calP{\mathcal{P}}
\newcommand\calQ{\mathcal{Q}}
\newcommand\calU{\mathcal{U}}

\newcommand\bbR{\mathbb{R}}

\newcommand\frakg{\mathfrak{g}}
\newcommand\frakh{\mathfrak{h}}
\newcommand\frakm{\mathfrak{m}}
\newcommand\frakp{\mathfrak{p}}

\newcommand\frakso{\mathfrak{so}}

\newcommand\frakX{\mathfrak{X}}
\newcommand\frakY{\mathfrak{Y}}
\newcommand\frakZ{\mathfrak{Z}}

\newcommand{\tv}{\tilde{v}}
\newcommand{\tw}{\tilde{w}}
\newcommand{\tvarpi}{\widetilde{\varpi}}

\newcommand{\tX}{\widetilde{X}}

\newcommand\hd{\widehat{\dd}} 
\newcommand{\hg}{\widehat{g}}
\newcommand{\hh}{\widehat{h}}
\newcommand{\hi}{\hat{\imath}}
\newcommand{\hj}{\hat{\jmath}}
\newcommand{\hr}{\widehat{r}}
\newcommand{\hs}{\hat{s}}
\newcommand{\hv}{\hat{v}}
\newcommand{\hw}{\widehat{w}}
\newcommand{\hS}{\widehat{S}}
\newcommand{\hX}{\widehat{X}}

\newcommand{\hfX}{\widehat{\frakX}}
\newcommand{\hfY}{\widehat{\frakY}}
\newcommand{\hgamma}{\widehat{\gamma}}

\newcommand{\homega}{\widehat{\omega}}

\newcommand{\hPsi}{\widehat{\Psi}}

\newcommand{\bh}{\bar{h}}
\newcommand{\bi}{\bar{\imath}}
\newcommand{\bv}{\bar{v}}
\newcommand{\bw}{\bar{w}}
\newcommand{\bA}{\bar{A}}
\newcommand{\bX}{\bar{X}}
\newcommand{\bgamma}{\bar{\gamma}}

\newcommand{\bOmega}{\overline{\Omega}}
\newcommand{\bnabla}{\bar{\nabla}}

\newcommand{\mrA}{\mathring{A}}

\newcommand{\mromega}{\mathring{\omega}}

\newcommand{\mrnabla}{\mathring{\nabla}}

\newcommand\loc{{\text{\textup{loc}}}}
\newcommand{\omegaloc}{\omega_\loc}

\newcommand\lie{{\text{\textup{Lie}}}}
\newcommand{\varpiLie}{\varpi_\lie}
\newcommand{\tvarpiLie}{\tvarpi_\lie}

\newcommand{\omegaLie}{\omega_\lie}
\newcommand{\OmegaLie}{\Omega_\lie}

\newcommand{\mromegaLie}{\mromega_\lie}

\newcommand{\betaLie}{\beta_\lie}
\newcommand{\bOmegaLie}{\bOmega_\lie}

\DeclareMathOperator{\Aut}{Aut}
\DeclareMathOperator{\Diff}{Diff}
\DeclareMathOperator{\Id}{Id}

\DeclareMathOperator{\im}{Im}

\newcommand{\Ad}{\text{\textup{Ad}}}
\newcommand{\ad}{\text{\textup{ad}}}

\newtheorem{theorem}{Theorem}[section]
\newtheorem{lemma}[theorem]{Lemma}
\newtheorem{corollary}[theorem]{Corollary}
\newtheorem{proposition}[theorem]{Proposition}
\theoremstyle{definition}

\newtheorem{remark}[theorem]{Remark}  

\hypersetup{
pdfauthor={J. Attard, J. François, S. Lazzarini, T. Masson},
pdftitle={Cartan Connections and Atiyah Lie Algebroids},
pdfsubject={},
pdfcreator={pdflatex},
pdfproducer={pdflatex},
pdfkeywords={}
}

\begin{document}

\title{Cartan Connections and Atiyah Lie Algebroids}

\author{J. Attard$^a$, J. François,$^b$ S. Lazzarini$^a$ and T. Masson$^a$\\[2.5mm]
{\normalsize $^a$ Centre de Physique Théorique}\\
{\normalsize Aix Marseille Univ, Université de Toulon, CNRS, CPT, Marseille, France.}\\[1mm]
{\normalsize $^b$ Service de Physique de l'Univers, Champs et Gravitation, Université de Mons – UMONS,}\\ 
{\normalsize 20 Place du Parc, B-7000 Mons, Belgique.}%
} 

\date{August 23, 2019}

\maketitle

\bigskip

\begin{center}\itshape
In memoriam of Sir Michael Atiyah
\end{center}

\bigskip

\begin{abstract}
This work extends both classical results on Atiyah Lie algebroids and previous developments carried out by some of the authors on Ehresmann connections on Atiyah Lie algebroids in their algebraic version. In this paper, we study Cartan connections in a framework relying on two Atiyah Lie algebroids based on a $H$-principal fiber bundle $\mathcal{P}$ and its associated $G$-principal fiber bundle $\mathcal{Q} := \mathcal{P} \times_H G$, where $H \subset G$ defines the model for a Cartan geometry. Completion of a known commutative and exact diagram relating these two Atiyah Lie algebroids allows to completely characterize Cartan connections on $\mathcal{P}$ as a fresh standpoint. Furthermore, in the context of gravity and mixed anomalies, our construction answers a long standing mathematical question about the correct geometrico-algebraic setting in which to combine inner gauge transformations and infinitesimal diffeomorphisms.
\end{abstract}



\smallskip
\noindent{\bfseries AMS classification scheme numbers:} 
51P05, 
70S15, 
70S20, 
83C99  

\newpage

\tableofcontents

\section{Introduction}
\label{sec-Intro}

Our general view on confirmed fundamental physical theories permits to distinguish two types of symmetries: on the one hand “outer” symmetries stemming from transformations of spacetime $\calM$, that is diffeomorphisms $\Diff(\calM)$; on the other hand “inner” symmetries stemming from the action of a gauge group $\calH$. Gauge field theories, in which inner symmetries play the essential role, can be formalized using at least three different mathematical frameworks: the usual geometry of fiber bundles, noncommutative geometry, and transitive Lie algebroids.  We refer to \cite{FranLazzMass14a} for a review and comparisons of the models that can be developed in these various approaches. These gauge theories rely on Ehresmann connections, those behind Yang-Mills fields in physics. On the other hand, formulations of gravity  and conformal theories are well encompassed in the use of Cartan connections.

In this paper, we propose to explore a formulation of Cartan connections in the framework of Atiyah Lie algebroids, which are special cases of transitive Lie algebroids. 

On the mathematical side, this study is motivated by our previous papers \cite{FourLazzMass13a,LazzMass12a} on generalization of Ehresmann connections on transitive Lie algebroids, to which we refer for details. In particular, in \cite{LazzMass12a} we were able to completely characterize Ehresmann connections on Atiyah Lie algebroids in their algebraic guise (using the language of modules of sections and spaces of maps rather than manifolds and bundles)
\begin{equation}
\label{eq ses Atiyah}
\begin{tikzcd}[column sep=25pt, row sep=30pt]
\algzero
	\arrow[r]
& \Gamma_H(\calP, \frakh)
	\arrow[r]
& \Gamma_H(\calP)
	\arrow[r]
& \Gamma(T\calM)
	\arrow[r]
& \algzero
\end{tikzcd}
\end{equation}
constructed on a $H$-principal fiber bundle $\calP$. Here, we do the same for Cartan connections, using a more involved construction, the commutative diagram~\eqref{eq-diagram}, which uses two  Atiyah Lie algebroids: the first one is \eqref{eq ses Atiyah} and the second one is based on the associated $G$-principal fiber bundle $\calQ \defeq \calP \times_H G$, where $H \subset G$ defines the model for a Cartan geometry \cite{Shar97a}. This extension of the $H$-principal fiber bundle $\calP$ to its associated $G$-principal fiber bundle $\calQ$ is an essential procedure\footnote{It is also referred to as an extension of the structural group or as an extension (or prolongation) of a principal bundle along a group homomorphism in the mathematical literature.} in our considerations: see \textsl{e.g.} \cite[V.2.5.5]{GreuHalpVans73a} for the general construction of a $\lambda$-extension (here $\lambda : H \to G$ is the inclusion map), see in addition~\cite[Appendix~A]{Mack87a} for its use in the context of Atiyah Lie algebroids, and \cite[Appendix~A.3]{Shar97a} for its use in the context of Cartan connections. One of our main results is Theorem~\ref{thm eq cartan iso} which identifies a Cartan connection as an isomorphism $\varpiLie : \Gamma_H(\calP) \to \Gamma_H(\calP, \frakg)$ satisfying a certain normalization condition \eqref{eq-diagram-triangle-upleft}.

As in \cite{LazzMass12a}, the notion of (infinitesimal) gauge transformations is studied in this framework. Here, we insist on the fact that Lie derivative along any element $\frakX \in \Gamma_H(\calP)$ combines an inner gauge transformation and an infinitesimal diffeomorphism, so taking into account, in a single structure, the main symmetries revealed by present fundamental physical theories.

On the physical side, this last point was a strong motivation for the present study. Indeed, it has been a long standing issue in field theories to propose a coherent and powerful mathematical setting in which it is possible to represent the combined action of infinitesimal gauge transformations $v \in \text{Lie}\,\calH$ (where $\calH$ is the group of vertical automorphisms of $\calP$) and infinitesimal diffeomorphisms $ X \in \Gamma(T\calM)\ (= \text{Lie}\Diff(\calM))$ on fields. As a case study, this problem is addressed in \cite{LangSchuStor84a, Stor86a, Stor93a, Stor06a}, using reductive Cartan connections in the context of gravity and mixed anomalies. In \cite{LangSchuStor84a} for instance, the action of $v$ can be directly implemented on fields, since $v$ is defined at the level of the principal fiber bundle $\calP$. But the action of $X$, defined at the level of the base manifold $\calM$, requires a lifting $X \mapsto \nabla_X \in \text{Lie}\Aut(\calP)$ to make it an object acting of fields at the level of $\calP$. This can be summarized in the following diagram, which is nothing but \eqref{eq ses Atiyah} since $\Gamma_H(\calP, \frakh) = \text{Lie}\,\calH$ ($\Ad$-$H$-equivariant maps $\calP \to \frakh = \text{Lie}\,H$) and $\Gamma_H(\calP) = \text{Lie}\Aut(\calP)$ (where $\Aut(\calP)$ is the group of bundle automorphisms of $\calP$):
\begin{equation*}
\begin{tikzcd}[column sep=25pt, row sep=30pt]
\algzero
	\arrow[r]
& \text{Lie}\,\calH
	\arrow[r]
& \text{Lie}\Aut(\calP)
	\arrow[r]
& \Gamma(T\calM)
	\arrow[r]
	\arrow[l, bend right=25, "\nabla"', dashed]
& \algzero\,.
\end{tikzcd}
\end{equation*}
According to Stora \cite{Stor86a}, in this construction, the “parallel transport” $\nabla$, even if “physically appealing”, has still an obscure mathematical status.\footnote{It is a long-standing issue about which one of us (S.L.) was asked about by Stora in summer 1991.} We explicitly show that the framework we develop and study in the present paper answers the questions raised in \cite{Stor86a}, and reiterated in \cite{Stor93a} for instance.

Let us make the following remarks. As in \cite{FourLazzMass13a,LazzMass12a}, we will consider Atiyah Lie algebroids in the language of modules over $C^\infty(\calM)$: this is more suited because physics requires the use of sections (namely fields) of vector bundles. In this paper, we shall be at first mainly concerned with the mathematical structures emerging from the general construction and subsequently more detailed physical applications will be addressed in Sections \ref{subsec-GRanomalies} and \ref{subsec Field theory and Lagrangian}.

\smallskip
The paper is organized as follows. In Section~\ref{sec General constructions} we recall some usual constructions related to Cartan geometry and to Atiyah Lie algebroids. One of the key ingredients for further technical developments is presented in Section~\ref{subsec-transport-connection-P-Q}.

Section~\ref{sec-exact-commutative-diagram} is devoted to the construction of the diagram~\eqref{eq-diagram}, which is one of our main results. It is an exact commutative diagram featuring the Atiyah Lie algebroids of $\calP$ and $\calQ$, and whose maps are essential to characterize Cartan connections on $\calP$. A part of this diagram can be found in \cite[Appendix~A]{Mack87a}, see Remark~\ref{rmk Mackenzie}. In Section~\ref{sec The diagram in a local trivialization}, the whole construction is recast in a local trivialization of $\calP$ in order to make contact with field theoretical constructions already obtained in the literature.

Then, in Section~\ref{sec Connections and metrics} we introduce, on top of diagram~\eqref{eq-diagram}, Ehresmann and Cartan connections, as well as metrics. Theorem~\ref{thm eq cartan iso} is the first main result to mention in this section, since it permits to understand how Cartan connections fit into the diagram~\eqref{eq-diagram}, see diagram~\eqref{eq-diagram-cartan} as a new way of considering Cartan connections. Then, Proposition~\ref{prop-lie-derivative-A-theta} is the second main result to mention, since it identifies Lie derivative along $\frakX \in \Gamma_H(\calP)$ as a combined inner gauge transformation and infinitesimal diffeomorphism.

Section~\ref{sec Applications} is devoted to some applications. Constructing Cartan connections in the framework of Lie algebroids is not  new. It has been explored for instance in \cite{CramSaun16a}, but from a different point of view. In particular the authors start from Lie groupoids, which is a very different approach to ours. However, a comparison is proposed in Section~\ref{sec comparison Crampin Saunder}. In Section~\ref{subsec-GRanomalies} we show that our construction provides a framework in which the above mentioned questions regarding gravity and mixed anomalies, as raised \textsl{e.g.} in \cite{LangSchuStor84a, Stor93a}, find a natural resolution as we easily reproduce particular results of \cite{LangSchuStor84a}.
Finally, the last two subsections \ref{sec generalized cartan connections} and \ref{subsec Field theory and Lagrangian} mention without details two questions flowing naturally from our construction, since they were addressed for Ehresmann connections in \cite{FourLazzMass13a,LazzMass12a}: the possible generalization of Cartan connections, and the construction of a Lagrangian for field theories based on Cartan connections. These problems are out of the scope of the present paper, and they will be the subject of further studies.

Even if Sections~\ref{sec General constructions} and \ref{sec-exact-commutative-diagram} contain some rather well-known results in classical differential geometry (bundle theory), full details will be provided in order to have a self-contained work and to promote the module viewpoint.

\section{General constructions}
\label{sec General constructions}

In this section we introduce some general constructions about Atiyah Lie algebroids (see \cite{FourLazzMass13a,LazzMass12a} for details) that will be used in the context of Cartan geometry (see \cite{Shar97a} for details).

\subsection{General definitions}

Let $\calM$ be an $n$-dimensional Riemannian manifold (or Lorentzian manifold of signature $(p, q)$) and $\Gamma(T\calM)$ the set of vector fields of $\calM$. Let $\calP=\calP(\calM, H)$ be a $H$-principal bundle over $\calM$, with $H \subset G$ acting by left multiplication on the group $G$. We denote by $L^G$ and $R^G$ the left and the right actions of $G$ on itself and by $\frakh$ and $\frakg$ the Lie algebras of $H$ and $G$.

\smallskip
Let $\calQ \defeq \calP \times_H G$ be the associated extension $G$-principal fiber bundle for the left action of $H$ on $G$. We denote by $R^\calP_h$ (resp. $R^\calQ_g$) the right action of $h \in H$ on $\calP$ (resp. $g \in G$ on $\calQ$). The fiber bundle $\calQ$ is defined by the projection $\Pi_\calQ : \calP \times G \to \calQ$ induced by the right action $\alpha_h(p,g) \defeq (ph, h^{-1} g)$ of $h \in H$. We use the notation $[p,g] \defeq \Pi_\calQ(p,g)$, and we have $[ph, h^{-1}g] = [p,g]$ for any $h \in H$, and $R^\calQ_{g'} [p,g] = [p, g g']$ for any $g' \in G$. The right action $R^\calQ_{g'} q = [p, g g']$ is induced by the right action $R^{\calP \times G}_{g'} (p,g) \defeq (p, g g')$ on $\calP \times G$, so that $R^\calQ_{g'} \circ \Pi_\calQ (p,g) = \Pi_\calQ \circ R^{\calP \times G}_{g'} (p,g)$. Notice that $R^{\calP \times G}_{g'} = \Id_\calP \times R^G_{g'}$.

The natural inclusion $\zeta : \calP \to \calQ$, $\zeta(p) \defeq [p,e]$ is compatible with the right actions of $H$ on $\calP$ and $\calQ$, since $[p,e]h = [p,h] = [ph,e]$. This inclusion permits to identify $\calP$ with $\zeta(\calP) = \{ [p,e] \mid p \in \calP \} \subset \calQ$.

Denote by $\Gamma_H(\calP, \frakh)$ the space of equivariant maps $v : \calP \to \frakh$ satisfying $({R^\calP_h}^* v)(p) = v(ph) = \Ad_{h^{-1}} (v(p))$ for any $p \in \calP$ and $h \in H$. In the same way, denote by $\Gamma_G(\calQ, \frakg)$ the space of equivariant maps $\calQ \to \frakg$. We introduce the associated vector bundles $\calL^\calP \defeq \calP \times_{\Ad} \frakh$ and $\calL^\calQ \defeq \calQ \times_{\Ad} \frakg$: then we have the well-known isomorphisms between sections of associated bundles and equivariant maps $\Gamma(\calL^\calP) \simeq \Gamma_H(\calP, \frakh)$ and $\Gamma(\calL^\calQ) \simeq \Gamma_G(\calQ, \frakg)$.

Let $\Gamma_H(\calP) \subset \Gamma(T\calP)$ be the subspace of right invariant vector fields on $\calP$: $T_p R^\calP_h(\frakX_{| p}) = \frakX_{| ph}$ for any $p \in \calP$, $h \in H$ and $\frakX \in \Gamma_H(\calP)$. In the same way, let $\Gamma_G(\calQ)$ be the space of right invariant vector fields on $\calQ$: $T_q R^\calQ_g(\hfX_{| q}) = \hfX_{| qg}$ for any $q \in \calQ$, $g \in G$ and $\hfX \in \Gamma_G(\calQ)$. These spaces define the (transitive) Atiyah Lie algebroids associated to $\calP$ and $\calQ$, and they give rise to the short exact sequence of Lie algebras and $C^\infty(\calM)$-modules
\begin{equation}
\label{eq-sec-Atiyah-P}
\begin{tikzcd}[column sep=25pt, row sep=30pt]
\algzero
	\arrow[r]
& \Gamma_H(\calP, \frakh)
	\arrow[r, "\iota_\calP"]
& \Gamma_H(\calP)
	\arrow[r, "\rho_\calP"]
& \Gamma(T\calM)
	\arrow[r]
& \algzero
\end{tikzcd}
\end{equation}
and
\begin{equation}
\label{eq-sec-Atiyah-Q}
\begin{tikzcd}[column sep=25pt, row sep=30pt]
\algzero
	\arrow[r]
& \Gamma_G(\calQ, \frakg)
	\arrow[r, "\iota_\calQ"]
& \Gamma_G(\calQ)
	\arrow[r, "\rho_\calQ"]
& \Gamma(T\calM)
	\arrow[r]
& \algzero
\end{tikzcd}
\end{equation}
In this paper, we use the identification $C^\infty(\calM) \simeq C_H^\infty(\calP) \simeq C_G^\infty(\calQ)$, where $C_H^\infty(\calP) =\{ f \in C^\infty(\calP) \mid f(ph)=f(p) \}$ and $C_G^\infty(\calQ)$ defined likewise. For an approach of Atiyah Lie algebroids in the standard language of bundles, see the classical references \cite{Mack87a, Mack05a}. As already stated, we prefer to use the language of sections (mainly as equivariant maps on principal fiber bundles) in order to be as close as possible to the field theoretical approach for the purpose of physical applications.

\smallskip
The brackets on $\Gamma_H(\calP, \frakh)$ and $\Gamma_G(\calQ, \frakg)$ are induced pointwise by the brackets in $\frakh$ and $\frakg$, and the brackets on $\Gamma_H(\calP)$ and $\Gamma_G(\calQ)$ are induced by the brackets of vector fields. The map $\iota_\calP$ is defined by 
\begin{equation*}
\iota_\calP(v)(p) \defeq -v(p)^v_{| p} = \frac{d}{dt} p \cdot e^{-t v(p)} {}_{| t=0}
\end{equation*}
where, for any $\eta \in \frakh$, $\eta^v$ denotes the vertical vector field on $\calP$ associated to the right action $R^\calP$. The map $\rho_\calP$ is the projection of right invariant vector fields on $\calP$ to vector fields on $\calM$. The maps $\iota_\calQ$ and $\rho_\calQ$ are defined in the same way. We refer to \cite{LazzMass12a} for more details and references.

The spaces $\Gamma_H(\calP, \frakh)$ and $\Gamma_G(\calQ, \frakg)$ are the so-called kernels of these transitive Lie algebroids. Denote by $\Gamma_H(\calP, \frakg)$ the space of $H$-equivariant maps $\calP \to \frakg$. Then the natural inclusion $i_\frakh : \frakh \hookrightarrow \frakg$ induces a natural inclusion $i : \Gamma_H(\calP, \frakh) \to \Gamma_H(\calP, \frakg)$ which is an injection and a morphism of Lie algebras and $C^\infty(\calM)$-modules. Since this map is “just an inclusion”, it will be often omitted in some expressions.

\begin{proposition}
There is a natural injective morphism of Lie algebras and $C^\infty(\calM)$-modules
\begin{equation}
\label{eq-def-j}
j: \Gamma_H(\calP, \frakh) \to \Gamma_G(\calQ, \frakg),
\quad 
j(v)([p,g]) \defeq \Ad_{g^{-1}} \circ i \circ v(p).
\end{equation}
\end{proposition}

\begin{proof}
First, for any $h \in H$, we have $j(v)([ph, h^{-1}g]) = \Ad_{g^{-1}h} \circ i \circ v(ph) = \Ad_{g^{-1}} \circ \Ad_{h} \circ  i \circ \Ad_{h^{-1}} (v(p)) = j(v)([p,g])$ since $i$ and $\Ad_{h^{-1}}$ commute. So $j(v)$ is well defined on $\calQ$.

Let us show that $j(v) :  \calQ \to \frakg$ is $G$-equivariant. For any $q=[p,g] \in \calQ$ and $g' \in G$,
\begin{align*}
j(v)(q g')
&= j(v)([p, g g'])
= \Ad_{(g g')^{-1}} \circ i \circ v(p)
= \Ad_{g'^{-1}} \circ \Ad_{g^{-1}} \circ i \circ v(p)
= \Ad_{g'^{-1}} \circ j(v)([p, g])
= \Ad_{g'^{-1}} \circ j(v)(q).
\end{align*}

The maps $\Ad_{g^{-1}}$ and $i$ are morphisms of Lie algebras and commute with multiplication by elements in $C^\infty(\calM) \simeq C_H^\infty(\calP) \simeq C_G^\infty(\calQ)$, so $j$ is a morphism of $C^\infty(\calM)$-modules and Lie algebras.

If $j(v)(q) = 0$ for any $q \in \calQ$, then $\Ad_{g^{-1}} \circ i \circ v(p) = 0$ for any $(p,g) \in \calP \times G$, and so $v(p) = 0$ for any $p \in \calP$ (since $i$ is injective), which proves the injectivity of $j$.
\end{proof}

Since $\calQ$ is the quotient of $\calP \times G$ by the action $\alpha$, many structures and objects defined on $\calQ$ can be obtained as $\alpha$-invariant structures and objects on $\calP \times G$. For instance, for any $v \in \Gamma_H(\calP, \frakh)$, define 
\begin{equation*}
\hj(v) : P \times G \to \frakg \quad \text{by} \quad \hj(v)(p,g) \defeq \Ad_{g^{-1}} \circ i \circ v(p). 
\end{equation*}
Then $(\alpha_h^* \hj(v))(p,g) = \hj(v)(ph, h^{-1}g) = \Ad_{g^{-1}} \circ \Ad_{h} \circ i \circ v(ph) = \Ad_{g^{-1}} \circ i \circ v(p) = \hj(v)(p,g)$ so that $\hj(v)$ is $\alpha$-invariant. Moreover, $({R^{\calP \times G}_{g'}}^* \hj(v))(p,g) = \hj(v)(p, g g') = \Ad_{g'^{-1}} \circ \Ad_{g^{-1}} \circ i \circ v(p) = \Ad_{g'^{-1}} \circ \hj(v)(p,g)$ so that $\hj(v)$ is $R^{\calP \times G}$-equivariant. These two properties show that  $\hj(v)$ defines a $R^{\calQ}$-equivariant map $\calQ \to \frakg$, which is $j(v)$ by direct inspection.

Here are some useful results about structures defined on $\calQ$ via $\calP \times G$.
\begin{lemma}
A vector field $\tX \oplus \xi \in \Gamma(T\calP \oplus TG) = \Gamma(T(\calP \times G))$ induces a well defined vector field on $\calQ$ if and only if 
\begin{equation}
\label{eq-tXxi-Q}
T_{(p, g)} \Pi_\calQ (\tX_{| p} \oplus \xi_{| g}) =  T_{(ph, h^{-1} g)} \Pi_\calQ (\tX_{| ph} \oplus \xi_{| h^{-1} g}) \in T_{| [p, g]}\calQ
\end{equation}
for any $p\in \calP$, $h \in H$ and $g \in G$.
\end{lemma}

\begin{proof}
The vector field $\tX \oplus \xi$ induces the well defined vector field $q \mapsto \hX_{| q} \defeq T_{(p, g)} \Pi_\calQ (\tX_{| p} \oplus \xi_{| g}) \in T_q \calQ$ which does not depend on the representative $(p,g) \in \calP \times G$ of $q =[p,g] \in \calQ$ by \eqref{eq-tXxi-Q}.
\end{proof}

\begin{lemma}
\label{lem-basic-forms-Q}
A form $\homega$ on $\calP \times G$ induces a form $\omega^\calQ$ on $\calQ$ if and only if it is $\alpha$-basic, which means that $\alpha_h^* \homega = \homega$ for any $h \in H$ ($H$-invariance) and $i_\eta \homega = 0$ for any $\eta \in \frakh$ ($H$-horizontality), where  $i_\eta$ is the inner contraction by the fundamental vector field $(p,g) \mapsto \tX_{| p} \oplus \xi_{| g} \in T_p \calP \oplus T_g G = T_{(p,g)} (\calP \times G)$ associated to the action $\alpha$:
\begin{equation}
\label{eq-fundamental-vector-alpha}
\tX_{| p} \oplus \xi_{| g} 
\defeq 
\frac{d}{dt} (p \cdot e^{t\eta}, e^{-t\eta} g) {}_{| t=0}
= \frac{d}{dt} p \cdot e^{t\eta} {}_{| t=0} \oplus  \frac{d}{dt} e^{-t\eta} g {}_{| t=0}
= \eta^v_{| p} \oplus -T_e R^G_g (\eta).
\end{equation}
\end{lemma}

\begin{proof}
Consider $\calP \times G$ as a $H$-principal fiber bundle over $\calQ$ and apply \cite[Sect.~6.3]{GreuHalpVans73a}.
\end{proof}

\begin{lemma}
\label{lem-iotaQv-TLGv}
Let $v \in \Gamma_G(\calQ, \frakg)$, then, $(p,g) \mapsto 0_{| p} \oplus - T_e L^G_{g} \left( v(q) \right) \in T_{(p,g)} (\calP \times G)$ satisfies \eqref{eq-tXxi-Q} and induces $\iota_\calQ(v) \in \Gamma_G(\calQ)$: for any $q = [p,g] \in \calQ$, one has
\begin{equation}
\label{eq-iotaQv-TLv}
\iota_\calQ(v)_{| q} 
=
T_{(p,g)} \Pi_\calQ \left( 0_{| p} \oplus - T_e L^G_{g} \left( v(q) \right) \right).
\end{equation}
\end{lemma}

\begin{proof}
One has $ - T_e L^G_{g} \left( v(q) \right) = \frac{d}{dt} g e^{-t v(q)} {}_{| t=0}$ so that
\begin{align*}
T_{(ph , h^{-1} g)} \Pi_\calQ \left( 0_{| ph} \oplus - T_e L^G_{h^{-1}g} ( v([ph, h^{-1}g]) ) \right)
&= \frac{d}{dt} \Pi_\calQ \left( ph,  h^{-1} g e^{-t v([ph, h^{-1} g])}\right)_{| t=0}
= \frac{d}{dt} \Pi_\calQ \left( p,  g e^{-t v([p, g])}\right)_{| t=0}
\\
&= 0_{| p} \oplus - T_e L^G_{g} \left( v(q) \right).
\end{align*}
By a direct computation, one has:
\begin{align*}
\iota_\calQ(v)_{| q}
&= \frac{d}{dt} [p,g] e^{-t v(q)} {}_{| t=0}
= \frac{d}{dt} \Pi_\calQ \left( p, g e^{-t v(q)} \right)_{| t=0}
= T_{(p,g)} \Pi_\calQ \left( 0_{| p} \oplus - T_e L^G_{g} \left( v(q) \right)  \right).
\end{align*}
\end{proof}

\begin{lemma}
\label{lem-iotaQj-iotaP}
Let $v \in \Gamma_H(\calP, \frakh)$, then, for any $q = [p,g] \in \calQ$, one has
\begin{equation}
\label{eq-iotaQjv-iotaPv}
\iota_\calQ \circ j(v)_{| q} 
=
T_{(p,g)} \Pi_\calQ \left( \iota_\calP(v)_{| p} \oplus 0_{| g} \right).
\end{equation}
\end{lemma}

\begin{proof}
From Lemma~\ref{lem-iotaQv-TLGv}, we have $\iota_\calQ \circ j(v)_{| q} = T_{(p,g)} \Pi_\calQ \left( 0_{| p} \oplus - T_e L^G_{g} \left( j(v)(q) \right) \right)$ for $q=[p,g]$. Since $- T_e L^G_{g} \left( j(v)(q) \right) = - T_e L^G_{g} \circ \Ad_{g^{-1}} \circ v(p) = \frac{d}{dt} e^{-t v(p)} g {}_{| t=0}$, we get
\begin{align*}
T_{(p,g)} \Pi_\calQ \left( 0_{| p} \oplus - T_e L^G_{g} \left( j(v)(q) \right) \right)
&= \frac{d}{dt} \Pi_\calQ \left( p, e^{-t v(p)} g \right)_{| t=0}
= \frac{d}{dt} \Pi_\calQ \left( pe^{-t v(p)}, g \right)_{| t=0}
= T_{(p,g)} \Pi_\calQ \left( \iota_\calP(v)_{| p} \oplus 0_{| g} \right).
\end{align*} 
\end{proof}

By comparing \eqref{eq-iotaQv-TLv} and \eqref{eq-iotaQjv-iotaPv}, we see that since $v$ takes values in $\frakh$ and appears on the RHS of both the equations either as a vector field along $G$ or as a vector field along~$\calP$.

\subsection{Differential calculi}
\label{sec-differential-calculi}

Differential calculi on transitive Lie algebroids have been described in \cite[Sect.~3]{LazzMass12a}, to which we refer for more details. In the following, we restrict ourselves to forms on $\Gamma_H(\calP)$ with values in $\Gamma_H(\calP, \frakh)$ or $\Gamma_H(\calP, \frakg)$. In order to define the differential, we use the natural representation (of Lie algebroids) of $\Gamma_H(\calP)$ on $\Gamma_H(\calP, \frakh)$: for any $\frakX \in \Gamma_H(\calP)$ and $v \in \Gamma_H(\calP, \frakh)$, $\frakX \cdot v \in \Gamma_H(\calP, \frakh)$ is the action of the vector field $\frakX$ on the map $v$. Using the right invariance of $\frakX$ and the $H$-equivariance of $v$, this defines a $H$-equivariant map. In the same way, for $\hv \in \Gamma_H(\calP, \frakg)$, $\frakX \cdot \hv \in \Gamma_H(\calP, \frakg)$. 

We denote by $(\Omega^\grast_\lie(\calP, \frakh), \hd_\frakh)$ the graded space of forms on $\Gamma_H(\calP)$ with values in $\Gamma_H(\calP, \frakh)$ equipped with the differential
\begin{multline}
\label{eq-differentialAtiyahPh}
(\hd_\frakh \omega)(\frakX_1, \dots, \frakX_{p+1}) = \sum_{i=1}^{p+1} (-1)^{i+1} \frakX_i \cdot \omega(\frakX_1, \ldots \omi{i} \ldots, \frakX_{p+1})\\
+ \sum_{1 \leq i < j \leq p+1} (-1)^{i+j} \omega([\frakX_i, \frakX_j], \frakX_1, \ldots \omi{i} \ldots \omi{j} \ldots, \frakX_{p+1})
\end{multline}
for any $\omega \in \Omega^p_\lie(\calP, \frakh)$. In the same way, we denote by $(\Omega^\grast_\lie(\calP, \frakg), \hd)$ the graded space of forms with values in $\Gamma_H(\calP, \frakg)$ equipped with a differential $\hd$ defined with the same formula. The spaces $\Omega^\grast_\lie(\calP, \frakh)$ and $\Omega^\grast_\lie(\calP, \frakg)$ are naturally graded Lie algebras and $C^\infty(\calM)$-modules.

For any $\frakX \in \Gamma_H(\calP)$, it is natural to consider the inner operation $i_\frakX$ on $\Omega^\grast_\lie(\calP, \frakh)$ or $\Omega^\grast_\lie(\calP, \frakg)$, which consists in saturating the first argument of a form with $\frakX$. Then we can define the Lie derivatives along $\frakX$ as
\begin{align}
\label{eq-def-lie-derivative}
L_\frakX &\defeq i_\frakX \, \hd_\frakh + \hd_\frakh \, i_\frakX \text{ on $\Omega^\grast_\lie(\calP, \frakh)$},
&
L_\frakX &\defeq i_\frakX \, \hd + \hd \, i_\frakX \text{ on $\Omega^\grast_\lie(\calP, \frakg)$}.
\end{align}

\begin{lemma}
\label{lem-actions-vector-fields}
For any $\frakX \in \Gamma_H(\calP)$, $v, v' \in \Gamma_H(\calP, \frakh)$, $\hw \in \Gamma_H(\calP, \frakg)$, one has
\begin{align}
\label{eq-actions-vector-fields}
[\frakX, \iota_\calP(v)] &= \iota_\calP(\frakX \cdot v),
&
\frakX \cdot i(v) & = i( \frakX \cdot v),
&
\iota_\calP(v) \cdot v' &= [v, v'],
&
\iota_\calP(v) \cdot \hw &= [i(v), \hw].
\end{align}
\end{lemma}

\begin{proof}
The first relation is proved in \cite[Sect.~4.2]{LazzMass12a} and the second relation is a direct consequence of the definition of $i$ from the injection $i_\frakh$. The third and fourth relations are proved with similar computations:
\begin{align*}
(\iota_\calP(v) \cdot v')(p)
&= \frac{d}{dt} v'\left( p\, e^{-t v(p)} \right)_{| t=0}
= \frac{d}{dt} \Ad_{e^{t v(p)}} (v'(p)) {}_{| t=0}
= [v(p), v'(p)],
\\
(\iota_\calP(v) \cdot \hw)(p)
&= \frac{d}{dt} \hw\left( p\, e^{-t v(p)} \right)_{| t=0}
= \frac{d}{dt} \Ad_{e^{t i(v(p))}} (\hw(p)) {}_{| t=0}
= [i(v(p)), \hw(p)].
\end{align*}
\end{proof}

\subsection{Local trivializations}
\label{sec-local-trivializations}

Trivialization of structures on Atiyah Lie algebroids are performed in different steps, see \cite{LazzMass12a, FourLazzMass13a} to which we refer for details. A local trivialization of $\calP$ relies on the choice of a local section $s : \calU \to \calP_{| \calU}$. In order to trivialize $\calQ$ at the same time and in a compatible way, we define the (canonically) associated local section of $\calQ$ by
\begin{equation*}
\hs : \calU \to \calQ_{| \calU}, \quad \hs(x) \defeq \Pi_\calQ (s(x), e) = [s(x), e].
\end{equation*}
Throughout the paper, we shall use a wavy arrow $\locto$ to relate a (global) space or object on $\calP$ or $\calQ$ to its local version over $\calU$. For instance, $\calP_{| \calU} \locto \calU \times H$ with $p \locto (x, h)$ such that $p = s(x) h$, and $\calQ_{| \calU} \locto \calU \times G$ with $q \locto (x, g)$ such that $q = \hs(x) g$.

The first easy step is to trivialize spaces $\Gamma_H(\calP, \frakm)$ of $H$-equivariant maps from $\calP$ to a $H$-module $\frakm$. Denote by $\rho(h) m$ the left action of $H$ on $\frakm$. Then one has $\Gamma_H(\calP_{| \calU}, \frakm) \locto \Gamma(\calU \times \frakm)$, where $v \locto \gamma \defeq s^* v : \calU \to \frakm$. The reconstruction map $\Psi_\frakm : \Gamma(\calU \times \frakm) \xrightarrow{\simeq} \Gamma_H(\calP_{| \calU}, \frakm)$ associates to $\gamma : \calU \to \frakm$ the map $\Psi_\frakm(\gamma)(p) \defeq \rho(h^{-1}) \gamma(x)$ for any $p = s(x) h \in \calP_{| \calU}$ and this map satisfies the $H$-equivariance relation $\Psi_\frakm(\gamma)(ph) = \rho(h^{-1}) \Psi_\frakm(\gamma)(p)$.

Using this procedure with $\rho = \Ad$, we get $\Gamma_H(\calP_{| \calU}, \frakh) \locto \Gamma(\calU \times \frakh)$, $\Gamma_H(\calP_{| \calU}, \frakg) \locto \Gamma(\calU \times \frakg)$, and $\Gamma_G(\calQ_{| \calU}, \frakg) \locto \Gamma(\calU \times \frakg)$ where for this last trivialization we used the pull-back by $\hs$. We denote by $\Psi_\frakh$, $\Psi_\frakg$, and $\hPsi_\frakg$ the corresponding reconstruction maps.

The next step is the trivialization $\Gamma_H(\calP_{| \calU}) \locto \Gamma(T\calU \oplus \calU \times \frakh)$. It is realized as follows. For any $\frakX \in \Gamma_H(\calP)$, define $X = \rho_\calP(\frakX) \in \Gamma(T\calM)$. Then, for $p = s(x) h$, $\bX_{| p} \defeq T_{s(x)} R^\calP_h \circ T_x s (X_{| x})$ defines a element in $\Gamma_H(\calP_{| \calU})$ and $\rho_\calP(\bX) = X$ by construction. This implies that there is a unique $v \in \Gamma_H(\calP_{| \calU}, \frakh)$ such that $\frakX = \bX + \iota_\calP(v)$. Then $\frakX \locto X \oplus \gamma$ where $X$ is restricted to $\calU$ and $\gamma \defeq s^* v$ as before. Denote by $S : \Gamma(T\calU \oplus \calU \times \frakh) \xrightarrow{\simeq} \Gamma_H(\calP_{| \calU})$ the reconstruction map defined by $S(X \oplus \gamma)_{| p} \defeq T_{s(x)} R^\calP_h \circ T_x s (X_{| x}) - [\Ad_{h^{-1}} \circ \gamma(x)]^v_{| s(x) h}$ for any $p=s(x)h$. Let us define $\mrnabla_{X | p} \defeq T_{s(x)} R^\calP_h \circ T_x s (X_{| x})$, then
\begin{equation}
\label{eq-S-notnabla-Psi}
S(X \oplus \gamma) = \mrnabla_X + \iota_\calP \circ \Psi_\frakh(\gamma).
\end{equation}
Notice that, with $\frakY \locto Y \oplus \eta$, one has $[\frakX, \frakY] \locto [X,Y] \oplus (X\cdot \eta - Y \cdot \gamma + [\gamma, \eta]) \rdefeq [X \oplus \gamma, Y \oplus \eta]$ and $S([X \oplus \gamma, Y \oplus \eta]) = [S(X \oplus \gamma), S(Y \oplus \eta)]_{| \calU}$.

Using the same procedure with $\rho_\calQ$, $\hs$, and $R^\calQ$, one has $\Gamma_G(\calQ_{| \calU}) \locto \Gamma(T\calU \oplus  \calU \times \frakg)$ with $\hfX \locto X \oplus \hgamma$. We denote by $\hS : \Gamma(T\calU \oplus  \calU \times \frakg) \xrightarrow{\simeq} \Gamma_G(\calQ_{| \calU})$ the reconstruction map.

Finally, let us describe trivializations of forms. For any $H$-module $\frakm$, define the graded space
\begin{align*}
\Omega^\grast_\lie(\calU, \frakh ; \frakm) \defeq \Omega^\grast(\calU) \otimes \exter^\grast \frakh^* \otimes \frakm
\end{align*}
where on the LHS the grading is the total grading of the RHS. We equip this space with the differential operators
\begin{align*}
\dd : \Omega^\grast(\calU) \otimes \exter^\grast \frakh^\ast \otimes \frakm
& \to \Omega^{\grast+1}(\calU) \otimes \exter^\grast \frakh^\ast \otimes \frakm,
&
\ds'_\frakm : \Omega^\grast(\calU) \otimes \exter^\grast \frakh^\ast \otimes \frakm
&\to \Omega^\grast(\calU) \otimes \exter^{\grast+1} \frakh^\ast \otimes \frakm,
\end{align*}
where $\dd$ is the de~Rham differential on $\Omega^\grast(\calU)$, and $\ds'_\frakm$ is the Chevalley--Eilenberg differential on $\exter^\grast \frakh^\ast \otimes \frakm$.

Then, one has
\begin{align*}
(\Omega^\grast_\lie(\calP_{| \calU}, \frakh), \hd_\frakh) 
&\locto
(\Omega^\grast_\lie(\calU, \frakh ; \frakh), \dd + \ds'_\frakh),
&
(\Omega^\grast_\lie(\calP_{| \calU}, \frakh), \hd) 
&\locto
(\Omega^\grast_\lie(\calU, \frakh ; \frakg), \dd + \ds'_\frakg),
\end{align*}
where a form $\omega \in \Omega^\grast_\lie(\calP_{| \calU}, \frakh)$ is trivialized as $\omegaloc \defeq \Psi^{-1}_\frakh \circ \omega \circ S$ where $S$ is applied to all the arguments of $\omega$. In the same way, for any $\omega \in \Omega^\grast_\lie(\calP_{| \calU}, \frakg)$, one has $\omega \locto \omegaloc \defeq \Psi^{-1}_\frakg \circ \omega \circ S$. These trivializations are isomorphisms of graded differential algebras, graded Lie algebras and $C^\infty(\calU)$-modules. For instance, $\hd_\frakh \omega \locto (\dd + \ds'_\frakh) \omegaloc$ for any $\omega \in \Omega^\grast_\lie(\calP_{| \calU}, \frakh)$. 

We refer to \cite{LazzMass12a, FourLazzMass13a} for details about the change of trivializations when one changes the trivialization section $s$ and/or the open subset $\calU$.

\subsection{\texorpdfstring{Connections on $\calP$ and transported connections on $\calQ$}{Connections on P and transported connections on Q}}
\label{subsec-transport-connection-P-Q}

Let us summarize some constructions and results from \cite{LazzMass12a}. To any Ehresmann connection $\nabla^\calP : \Gamma(T\calM) \to \Gamma_H(\calP)$, we associate a unique map $\alpha^\calP : \Gamma_H(\calP) \to \Gamma_H(\calP, \frakh)$ defined by the relation $\frakX = \nabla^\calP_X - \iota_\calP \circ \alpha^\calP(\frakX)$ for any $\frakX \in \Gamma_H(\calP)$ and where $X \defeq \rho_\calP(\frakX)$. It satisfies the normalization condition $\alpha^\calP \circ \iota_\calP = - \Id$ (and accordingly $\alpha^\calP \circ \nabla^\calP = 0$). There is a one-to-one correspondence between $1$-forms $\alpha^\calP \in \Omega^1_\lie(\calP, \frakh)$ satisfying $\alpha^\calP \circ \iota_\calP = - \Id$ and connections on $\calP$ \cite[Prop.~3.9]{LazzMass12a}. In other words, using $\nabla^\calP$, any $\frakX \in \Gamma_H(\calP)$ uniquely defines two global objects, namely $X\defeq \rho_\calP(\frakX)\in \Gamma(T\calM)$ and $v\defeq -\alpha^\calP(\frakX) \in \Gamma_H(\calP, \frakh)$, such that $\frakX = \nabla^\calP_X + \iota_\calP(v)$.

\medskip
The following technical considerations will be used to simplify the presentation of subsequent results in Sect.~\ref{sec-exact-commutative-diagram}.

\begin{proposition}
\label{prop-connextion-P-to-Q}
Let $\omega^\calP \in \Omega^1(\calP) \otimes \frakh$ be a connection $1$-form on $\calP$. Then the $1$-form defined on $\calP\times G$ by
\begin{equation}
\label{eq-def-omega-hat}
\homega_{| (p,g)} (\tX_{| p} \oplus \xi_{| g}) \defeq
\Ad_{g^{-1}} \left( \omega^\calP_{| p} (\tX_{| p}) \right) + T_g L^G_{g^{-1}} (\xi_{| g}),
\end{equation}
for any $(p,g) \in \calP \times G$ and $\tX_{| p} \oplus \xi_{| g} \in T_p \calP \oplus T_g G$, induces a ``transported'' connection $1$-form $\omega^\calQ$ on $\calQ$ by 
\begin{equation}
\label{eq-omegaQ-homega}
\omega^\calQ_{| q} (\hX_{| q}) \defeq \homega_{| (p,g)} (\tX_{| p} \oplus \xi_{| g}),
\end{equation}
for any $q = [p,g] \in \calQ$, $\hX_{| q} \in T_q \calQ$, and $\tX_{| p} \oplus \xi_{| g} \in T_p \calP \oplus T_g G$ such that $T_{(p,g)} \Pi_\calQ (\tX_{| p} \oplus \xi_{| g}) = \hX_{| q}$.
\end{proposition}

\begin{proof}
First, let us prove that $\homega$ is $H$-basic on $\calP \times G$. The $H$-invariance results from the computation:
\begin{align*}
(\alpha_h^* \homega)_{| (p,g)} (\tX_{| p} \oplus \xi_{| g})
&= \homega_{| (ph, h^{-1} g)} \left(  T_{(p,g)} \alpha_h (\tX_{| p} \oplus \xi_{| g}) \right)
= \homega_{| (ph, h^{-1} g)} \left( T_p R^\calP_h (\tX_{| p}) \oplus T_g L^G_{h^{-1}} (\xi_{| g}) \right)
\\
&= \Ad_{g^{-1} h} \left(  \omega^\calP_{ph} \circ T_p R^\calP_h (\tX_{| p})  \right)  + T_{h^{-1}g} L^G_{g^{-1}h} \circ T_g L^G_{h^{-1}} (\xi_{| g})
\\
&= \Ad_{g^{-1} h} \circ \Ad_{h^{-1}} \left( \omega^\calP_{| p} (\tX_{| p}) \right) + T_g L^G_{g^{-1}} (\xi_{| g})
= \homega_{| (p,g)} (\tX_{| p} \oplus \xi_{| g}).
\end{align*}
For the $H$-horizontality, let $\eta \in \frakh$ and let $\tX_{| p} \oplus \xi_{| g}$ be defined by \eqref{eq-fundamental-vector-alpha}. Using \eqref{eq-def-omega-hat} and $\omega^\calP(\eta^v) = \eta$, one has
\begin{align*}
\Ad_{g^{-1}} \left( \omega^\calP_{| p} (\tX_{| p}) \right) + T_g L^G_{g^{-1}} (\xi_{| g})
&=
\Ad_{g^{-1}}(\eta) - T_g L^G_{g^{-1}} \circ T_e R^G_g (\eta)
= 0.
\end{align*}
Since $\homega$ is $H$-basic on $\calP \times G$, it defines $\omega^\calQ \in \Omega^1(\calQ) \otimes \frakg$ by \eqref{eq-omegaQ-homega} (see Lemma~\ref{lem-basic-forms-Q}). It remains to show that $\omega^\calQ$ is a connection. For any $\hX_{| [p,g]} = T_{(p,g)} \Pi_\calQ \left( \tX_{| p} \oplus \xi_{| g} \right)$ one has
\begin{align*}
T_{[p,g]} R^\calQ_{g'} \left( \hX_{| [p,g]} \right)
&= T_{(p, gg')} \Pi_\calQ \circ T_{(p,g)} R^{\calP \times G}_{g'} \left( \tX_{| p} \oplus \xi_{| g} \right)
= T_{(p, gg')} \Pi_\calQ \left( \tX_{| p} \oplus T_g R^G_{g'} (\xi_{| g}) \right).
\end{align*}
This gives the right equivariance of $\omega^\calQ$:
\begin{align*}
\left ({R^\calQ_{g'}}^* \omega^\calQ \right)_{| [p,g]} \left( \hX_{| [p,g]} \right)
&=
\omega^\calQ_{| [p,g] g'} \left( T_{[p,g]} R^\calQ_{g'} (\hX_{| [p,g]}) \right)
= \homega_{| (p, gg')} \left( \tX_{| p} \oplus T_g R^G_{g'} (\xi_{| g}) \right)
\\
&= \Ad_{{g'}^{-1} g^{-1}}  \circ \omega^\calP_{| p} (\tX_{| p}) + T_{g g'} L^G_{{g'}^{-1} g^{-1}} \circ T_g R^G_{g'} (\xi_{| g})
\\
&= \Ad_{{g'}^{-1}} \left( \Ad_{g^{-1}} \circ \omega^\calP_{| p} (\tX_{| p}) + T_g L^G_{g^{-1}} (\xi_{| g}) \right)
\\
&= \Ad_{{g'}^{-1}} \circ \omega^\calQ_{[p,g]} \left( \hX_{| [p,g]} \right).
\end{align*}
For any $\eta \in \frakg$, the induced vertical vector field $\eta^v$ on $\calQ$ is given by
\begin{align*}
\eta^v_{| [p,g]} 
& \defeq \frac{d}{dt} [p,g] e^{t\eta} {}_{| t=0} 
= T_{(p,g)} \Pi_\calQ \left( \frac{d}{dt} (p, g e^{t\eta}) {}_{| t=0} \right)
= T_{(p,g)} \Pi_\calQ \left( 0_{| p} \oplus T_e L^G_g (\eta) \right),
\end{align*}
so that, by \eqref{eq-omegaQ-homega},
\begin{align*}
\omega^\calQ_{| [p,g]} \left( \eta^v_{| [p,g]}  \right) 
&= \homega_{(p,g)} \left( 0_{| p} \oplus T_e L^G_g (\eta) \right)
= T_g L^G_{g^{-1}} \circ T_e L^G_{g} (\eta)
= \eta,
\end{align*}
which proves the normalization of  $\omega^\calQ$ on vertical vector fields.
\end{proof}

Let us describe the result of Prop.~\ref{prop-connextion-P-to-Q} in terms of a local trivialization of $\calP$ and $\calQ$. We use the notations of Sect.~\ref{sec-local-trivializations}. For any $X_{| x} \in T_x \calM$, let
\begin{align*}
A^\calP(X)_{| x} 
&\defeq (s^* \omega^\calP)_{| x} (X_{| x})
= \omega^\calP_{| s(x)} \left( T_x s(X_{| x}) \right),
&
A^\calQ(X)_{| x} 
&\defeq (\hs^* \omega^\calQ)_{| x} (X_{| x})
= \omega^\calQ_{| \hs(x)} \left( T_x \hs(X_{| x}) \right),
\end{align*}
be the local trivializations of $\omega^\calP$ and $\omega^\calQ$ induced by $s$ and $\hs$.

\begin{lemma}
For $p = s(x) h$, $q = [p, g]$, and $X_{| x} \in T_x \calU$, one has
\begin{gather}
\label{eq-piQ-RP-s-RQ}
T_{(p, g)} \Pi_\calQ \left( T_{s(x)} R^\calP_h \circ T_x s (X_{| x}) \oplus 0_{| g} \right) 
= T_{\hs(x)} R^\calQ_{hg} \circ T_x \hs (X_{| x}),
\\
\label{eq-omegaQ-TRQ-Ths-AdomegaP-Ts}
\omega^\calQ_{| q} \left( T_{\hs(x)} R^\calQ_{hg} \circ T_x \hs (X_{| x}) \right) 
= \Ad_{(hg)^{-1}} \circ \omega^\calP_{| s(x)} \left( T_x s (X_{| x}) \right),
\\
\intertext{and}
\nonumber
A^\calQ = A^\calP \text{ that is } \hs^* \omega^\calQ = s^* \omega^\calP.
\end{gather}
\end{lemma}

Notice that with $h=g=e$, \eqref{eq-omegaQ-TRQ-Ths-AdomegaP-Ts} reduces to $\omega^\calQ_{| \hs(x)} \left( T_x \hs (X_{| x}) \right) = \omega^\calP_{| s(x)} \left( T_x s (X_{| x}) \right)$, which shows in particular that the LHS belongs to $\frakh$. This is obviously a consequence of the fact that the vertical component of $T_x \hs (X_{| x})$ takes its values in $\frakh$ only since $\hs$ can be considered as a local section of $\zeta(\calP) \subset \calQ$.

\begin{proof}
Let $(t,x) \mapsto \phi^X_t(x)$ such that $X_{| x} = \frac{d}{dt} \phi^X_t(x) {}_{| t=0}$. One has
\begin{align*}
T_{(p, g)} \Pi_\calQ \left( T_{s(x)} R^\calP_h \circ T_x s (X_{| x}) \oplus 0_{| g} \right)
&= \frac{d}{dt} \Pi_\calQ \left(  s(\phi^X_t(x)) h, g \right)_{| t=0}
= \frac{d}{dt} R^\calQ_{hg} \circ \Pi_\calQ \left( s(\phi^X_t(x)), e \right)_{| t=0}
\\
&
= T_{\hs(x)} R^\calQ_{hg} \circ T_x \hs (X_{| x}).
\end{align*}
Then, by \eqref{eq-omegaQ-homega},
\begin{align*}
\omega^\calQ_{| q} \left( T_{\hs(x)} R^\calQ_{hg} \circ T_x \hs (X_{| x}) \right)
&= \homega_{| (p,g)} \left( T_{s(x)} R^\calP_h \circ T_x s (X_{| x}) \oplus 0_{| g} \right)
= \Ad_{g^{-1}} \circ \omega^\calP_{| p} \left( T_{s(x)} R^\calP_h \circ T_x s (X_{| x}) \right)
\\
&
= \Ad_{g^{-1}} \circ \Ad_{h^{-1}} \circ \omega^\calP_{| s(x)} \left( T_x s (X_{| x}) \right).
\end{align*}
With $h=g=e$ in \eqref{eq-piQ-RP-s-RQ}, one gets 
\begin{equation*}
T_{\hs(x)} \Pi_\calQ \left( T_x s (X_{| x}) \oplus 0_{| e} \right) = T_x \hs (X_{| x}),
\end{equation*}
from which we deduce
\begin{align*}
A^\calQ(X)_{| x}
&= \homega_{(s(x), e)} \left( T_x s (X_{| x}) \oplus 0_{| e} \right)
= \omega^\calP_{| s(x)} \left( T_x s (X_{| x}) \right)
= A^\calP(X)_{| x}.
\end{align*}
\end{proof}

The connections $\omega^\calP$ and $\omega^\calQ$ define horizontal lifts $\nabla^\calP : \Gamma(T \calM) \to \Gamma_H(\calP)$ and $\nabla^\calQ : \Gamma(T \calM) \to \Gamma_G(\calQ)$ respectively. In the following, we will need some explicit expressions in a local trivialization for these lifts. Using the previous notations, one has the usual formula, for $p = s(x) h$, $q = \hs(x) g = [s(x), g]$, and $X \in \Gamma(T \calU)$,
\begin{align*}
\nabla^\calP_X {}_{| p} 
&= T_{s(x)} R^\calP_h \left[  T_x s (X_{| x}) - \left[ \omega^\calP_{| s(x)} \left( T_x s (X_{| x}) \right) \right]^v_{| s(x)} \right],
&
\nabla^\calQ_X {}_{| q} 
&= T_{\hs(x)} R^\calQ_g \left[  T_x \hs (X_{| x}) - \left[ \omega^\calQ_{| \hs(x)} \left( T_x \hs (X_{| x}) \right) \right]^v_{| \hs(x)} \right],
\end{align*}
where $\eta^v$ is the vertical vector field associated to $\eta \in \frakh$ either on $\calP$ or on $\calQ$.

\begin{lemma}
\label{lem-nablaQ-TpiQ-nablaP}
For any $q = [p,g] \in \calQ$ and $X \in \Gamma(T \calM)$, one has
\begin{equation*}
\nabla^\calQ_X {}_{| q}  = T_{(p, g)} \Pi_\calQ \left( \nabla^\calP_X {}_{| p} \oplus 0_{| g} \right)
\end{equation*}
\end{lemma}

\begin{proof}
Let us compute $T_{(p, g)} \Pi_\calQ$ on the two terms in $\nabla^\calP_X {}_{| p}$, and let us use $X_{| x} = \frac{d}{dt} \phi^X_t(x) {}_{| t=0}$ in the local trivializations as before, with $p = s(x) h$ and $q = \hs(x) hg$. For the first term, one has
\begin{align*}
T_{(p, g)} \Pi_\calQ \left( T_{s(x)} R^\calP_h \circ  T_x s (X_{| x}) \oplus 0_{| g} \right)
&= \frac{d}{dt} \Pi_\calQ \left( s(\phi^X_t(x)) h, g \right)_{| t=0}
= \frac{d}{dt} R^\calQ_{hg} \circ \Pi_\calQ \left( s(\phi^X_t(x)), e \right)_{| t=0}
\\
&= \frac{d}{dt} R^\calQ_{hg} \circ \hs(\phi^X_t(x)) {}_{| t=0} 
= T_{\hs(x)} R^\calQ_{hg} \circ T_x \hs (X_{| x}).
\end{align*}
For the second term, one has, using \eqref{eq-omegaQ-TRQ-Ths-AdomegaP-Ts} with $h=g=e$,
\begin{align*}
T_{(p, g)} \Pi_\calQ \left( T_{s(x)} R^\calP_h \left( \left[ \omega^\calP_{| s(x)} \left( T_x s (X_{| x}) \right) \right]^v_{| s(x)} \right) \oplus 0_{| g} \right)
& = \frac{d}{dt}  \Pi_\calQ \left( 
		R^\calP_h 
		\left( s(x) \exp \left( t \omega^\calP_{| s(x)} ( T_x s ( X_{| x} ) ) \right) \right)
		, g 
	\right)_{| t=0}
\\
& = \frac{d}{dt}  \Pi_\calQ \left( 
		s(x) \exp \left( t \omega^\calQ_{| \hs(x)} \left( T_x \hs (X_{| x}) \right) \right)
		, hg 
	\right)_{| t=0}
\\
&= \frac{d}{dt} \hs(x) \exp \left( t \omega^\calQ_{| \hs(x)} \left( T_x \hs (X_{| x}) \right) \right) hg {}_{| t=0}
\\
&= T_{\hs(x)} R^\calQ_{hg} \left( \left[ \omega^\calQ_{| \hs(x)} \left( T_x \hs (X_{| x}) \right) \right]^v_{| \hs(x)} \right).
\end{align*}
Combining the two terms with the correct $-$ sign, one gets the result.
\end{proof}

By definition, one has $\omega^\calP \circ \nabla^\calP = 0$ since the range of $\nabla^\calP$ is in the set of horizontal vector fields. From Lemma~\ref{lem-nablaQ-TpiQ-nablaP}, one gets $\omega^\calQ_{| q} \left( \nabla^\calQ_X {}_{| q} \right) = \homega_{| (p,g)} \left( \nabla^\calP_X {}_{| p} \oplus 0_{| g} \right) = \Ad_{g^{-1}} \circ \omega^\calP_{| p} \left(  \nabla^\calP_X {}_{| p} \right) = 0$ as expected.

\begin{corollary}
Any $\hfX \in \Gamma_G(\calQ)$ can be written as $\hfX = \nabla^\calQ_X + \iota_\calQ(v)$ for $X = \rho_\calQ(\hfX)$ and a unique $v \in \Gamma_G(\calQ, \frakg)$. Then the global vector field $\tX \oplus \xi$ on $\calP \times G$ defined by
\begin{equation*}
(\tX \oplus \xi)_{| (p,g)} \defeq \nabla^\calP_X {}_{| p} \oplus \left( - T_e L^G_g \left( v([p,g]) \right) \right)
\end{equation*}
satisfies
\begin{equation*}
{\Pi_\calQ}_* (\tX \oplus \xi) = \hfX.
\end{equation*}

\end{corollary}

\begin{proof}
See \cite{LazzMass12a} for the decomposition $\hfX = \nabla^\calQ_X + \iota_\calQ(v)$. Then Lemmas~\ref{lem-iotaQv-TLGv} and \ref{lem-nablaQ-TpiQ-nablaP} give the result.
\end{proof}

This corollary permits to perform computations with differential structures on $\calP \times G$ using objects like $\homega$ and $\tX \oplus \xi$ in place of computations on $\calQ$ with objects like $\omega^\calQ$ and $\hfX$.

\medskip
If $\omega'^\calP$ is another connection $1$-form on $\calP$, then there exists $\connexionShift \in \Omega^1(\calP) \otimes \frakh$, $R^\calP$-equivariant and $H$-horizontal, such that $\omega'^\calP = \omega^\calP + \connexionShift$. Then a direct computation shows that the connection $1$-form $\omega'^\calQ$ associated to $\omega'^\calP$ is given by, with $q = [p,g]$ and $\hX_{| [p,g]} = T_{(p,g)} \Pi_\calQ \left( \tX_{| p} \oplus \xi_{| g} \right)$,
\begin{equation*}
\omega'^\calQ_{| q} \left( \hX_{| q} \right) =  \omega^\calQ_{| q} \left( \hX_{| q} \right) + \Ad_{g^{-1}} \circ \connexionShift_{| p} \left( \tX_{| p} \right).
\end{equation*}

We can compare the lifts $\nabla^\calP$ and $\nabla'^\calP$ associated to the connections $\omega^\calP$ and $\omega'^\calP$. From $\omega'^\calP = \omega^\calP + \connexionShift$ one gets directly
\begin{align*}
\nabla'^\calP_X {}_{| p} = \nabla^\calP_X {}_{| p} - T_{s(x)} R^\calP_h \left( \left[ \connexionShift_{| s(x)} \left( T_x s (X_{| x}) \right) \right]^v_{| s(x)} \right)
\end{align*}
Using the $R^\calP$-equivariance of $\connexionShift$, the last term is
\begin{align*}
T_{s(x)} R^\calP_h \left( \left[ \connexionShift_{| s(x)} \left( T_x s (X_{| x}) \right) \right]^v_{| s(x)} \right)
&=  \frac{d}{dt} s(x) e^{t \connexionShift_{| s(x)} \left( T_x s (X_{| x}) \right)} h  {}_{| t=0} 
= \frac{d}{dt} s(x) h e^{t \Ad_{h^{-1}} \circ \connexionShift_{| s(x)} \left( T_x s (X_{| x}) \right)}  {}_{| t=0} 
\\
&
= \frac{d}{dt} p e^{t \connexionShift_{| p} \left( T_{s(x)} R^\calP_h \circ T_x s (X_{| x}) \right)}  {}_{| t=0} 
= \left[ \connexionShift_{| p} \left( T_{s(x)} R^\calP_h \circ T_x s (X_{| x}) \right) \right]^v_{| p}.
\end{align*}
The map 
\begin{equation*}
p \mapsto v^{X, \connexionShift} (p) \defeq \connexionShift_{| p} \left( T_{s(x)} R^\calP_h \circ T_x s (X_{| x}) \right) \in \frakh
\end{equation*}
 is defined here on $\calP_{| \calU}$ but it can be shown that it is a well defined map on $\calP$ which is $H$-equivariant: $v^{X, \connexionShift} \in \Gamma_H(\calP, \frakh)$. A way to look at it is to notice that since $\connexionShift \in \Omega^1(\calP) \otimes \frakh$ is a $R^\calP$-tensorial form, it defines a unique element $\hconnexionShift \in \Omega^1(\calM, \calL^\calP)$. Then, as a $R^\calP$-tensorial $0$-form, $v^{X, \connexionShift}$ is related to $\hconnexionShift(X) \in \Gamma(\calL^\calP)$ by the isomorphism $\Gamma(\calL^\calP) \simeq \Gamma_H(\calP, \frakh)$. Finally, one gets
\begin{align}
\label{eq-nablaprimeP-nablaP-vXshift}
\nabla'^\calP_X  = \nabla^\calP_X + \iota_\calP(v^{X, \connexionShift})
\end{align}

Then the lifts associated to $\omega^\calQ$ and $\omega'^\calQ$ are related in the same way: using \eqref{eq-omegaQ-TRQ-Ths-AdomegaP-Ts} with $h=g=e$ and $q = \hs(x)g$, 
\begin{align*}
\nabla'^\calQ_X {}_{| q} = \nabla^\calQ_X {}_{| q} - T_{\hs(x)} R^\calQ_g \left( \left[ \connexionShift_{| s(x)} \left( T_x s (X_{| x}) \right) \right]^v_{| \hs(x)} \right).
\end{align*}
The last term is 
\begin{align*}
- T_{\hs(x)} R^\calQ_g \left( \left[ \connexionShift_{| s(x)} \left( T_x s (X_{| x}) \right) \right]^v_{| \hs(x)} \right)
&= - T_{\hs(x)} R^\calQ_g \left( \left[ v^{X, \connexionShift}(s(x)) \right]^v_{| \hs(x)} \right)
= \frac{d}{dt} \hs(x) e^{- t v^{X, \connexionShift}(s(x))} g {}_{| t=0}
\\
&
= \frac{d}{dt} \hs(x) e^{- t \Ad_{g^{-1}} \circ v^{X, \connexionShift}(s(x))} {}_{| t=0}
= \iota_\calQ \circ j(v^{X, \connexionShift})_{| q} 
\end{align*}
where $j$ is defined in \eqref{eq-def-j}, so that
\begin{equation}
\label{eq-nablaprimeQ-nablaQ-vXshift}
\nabla'^\calQ_X = \nabla^\calQ_X + \iota_\calQ \circ j(v^{X, \connexionShift})
\end{equation}

Finally, let us recall \cite{LazzMass12a} that, using the notations of Sect.~\ref{sec-local-trivializations} and \ref{subsec-transport-connection-P-Q}, one has
\begin{align}\label{eq-trivial-X-v}
\nabla^\calP_X + \iota_\calP(v) &\locto X \oplus (A^\calP(X) + \gamma),
&
\nabla^\calQ_X + \iota_\calQ(\hv) &\locto X \oplus (A^\calQ(Y) + \hgamma),
\end{align}
where $v\locto \gamma$ and $\hv\locto \hgamma$.

\section{The exact commutative diagram}
\label{sec-exact-commutative-diagram}

One of the objectives of the forthcoming considerations in this section is to construct maps between the two transitive Lie algebroids \eqref{eq-sec-Atiyah-P} and \eqref{eq-sec-Atiyah-Q} which generalize $j$ defined in \eqref{eq-def-j}.

\subsection{\texorpdfstring{The short exact sequence induced by $j$}{The short exact sequence induced by j}}

\begin{proposition}
\label{prop-iso-gammaG-gammaH}
There is an isomorphism of Lie algebras and $C^\infty(\calM)$-modules 
\begin{align*}
\hi : \Gamma_H(\calP, \frakg) \xrightarrow{\simeq} \Gamma_G(\calQ, \frakg), 
\quad 
\hi(v)([p,g]) \defeq \Ad_{g^{-1}} \circ v(p).
\end{align*}
\end{proposition}

\begin{proof}
Given $v\in \Gamma_H(\calP, \frakg)$, one constructs the map $\vartheta: \calP \times G \to \frakg$ defined by $\vartheta(p, g) \defeq \Ad_{g^{-1}} \circ v(p)$. Then
\begin{align*}
(\alpha_h^* \vartheta)(p,g) 
&= \vartheta(ph, h^{-1} g)
= \Ad_{g^{-1} h} \circ v(ph)
= \Ad_{g^{-1} h} \circ \Ad_{h^{-1}} \circ v(p)
= \Ad_{g^{-1}} \circ v(p)
= \vartheta(p, g),
\end{align*}
and
\begin{align*}
({R^{\calP \times G}_{g'}}^* \vartheta)(p,g)
&= \vartheta(p,g g')
= \Ad_{{g'}^{-1} g^{-1}} \circ v(p)
= \Ad_{{g'}^{-1}} \circ \vartheta(p,g).
\end{align*}
So $\vartheta$ is $\alpha$-invariant and $R^{\calP \times G}$-equivariant and thus it defines a map $\hi(v) \in \Gamma_G(\calQ, \frakg)$ by $\hi(v)(q) \defeq \vartheta(p,g)$ for any $q = [p,g]$.

If $\hi(v)(q) = 0$ for any $q \in \calQ$, then $\Ad_{g^{-1}} \circ v(p) = 0$ for any $(p,g) \in \calP \times G$, and so $v(p) = 0$ for any $p \in \calP$, which proves the injectivity of $\hi$. Given $\hv \in \Gamma_G(\calQ, \frakg)$, let $v = \hv \circ \zeta$, \textit{i.e.} $v(p) \defeq \hv([p,e])$ for any $p \in \calP$. Then $v(ph) = \hv([p, e]h) = \Ad_{h^{-1}} \circ \hv([p, e]) = \Ad_{h^{-1}} \circ v(p)$, so that $v\in \Gamma_H(\calP, \frakg)$. It is straightforward to check that $\hi(v) = \hv$ since $\hi(v)([p,g]) = \Ad_{g^{-1}} \circ v(p) = \Ad_{g^{-1}} \circ \hv([p,e]) = \hv([p,g])$. This proves the surjectivity of $\hi$. 

Let us show that $\hi$ is a morphism of Lie algebras. Given $v, w \in \Gamma_H(\calP, \frakg)$, one has 
\begin{align*}
[\hi(v), \hi(w)]([p, g]) 
&= [\hi(v)([p, g]), \hi(w)([p, g])] 
= [\Ad_{g^{-1}} \circ v(p), \Ad_{g^{-1}} \circ w(p)] 
= \Ad_{g^{-1}} \circ [v, w](p)
\\
&
= \hi([v, w])([p,g]).
\end{align*}
The map $\hi$ is clearly a morphism of modules over the algebra $C^\infty(\calM) \simeq C_H^\infty(\calP) \simeq C_G^\infty(\calQ)$. 
\end{proof}

The isomorphism $\hi$ of Prop.~\ref{prop-iso-gammaG-gammaH} permits to define the injection
\begin{equation}
\label{eq-def-iota}
\iota : \Gamma_H(\calP, \frakh) \to \Gamma_G(\calQ), \quad \iota \defeq \iota_\calQ \circ \hi.
\end{equation}
and is related to $j : \Gamma_H(\calP, \frakh) \to \Gamma_G(\calQ, \frakg)$ defined by \eqref{eq-def-j} in the following way:
\begin{equation}
\label{eq-j-i-hi}
\begin{tikzcd}[column sep=25pt, row sep=30pt]
\Gamma_H(\calP, \frakh)
	\arrow[r, hookrightarrow, "i"]
	\arrow[rr, bend left=20, "j"]
& \Gamma_H(\calP, \frakg)
	\arrow[r, "\hi", "\simeq"']
& \Gamma_G(\calQ, \frakg)
\end{tikzcd}
\end{equation}
so that $j(v)([p, g]) = \hi \circ i(v) ([p, g])$.

For any $\xi \in \frakg$, we denote by $[\xi] \in \frakg /\frakh$ its (quotient) class. The vector space $\frakg/\frakh$ is a $H$-module for $\Ad_{h^{-1}}[\xi] \defeq [\Ad_{h^{-1}} \xi ]$. We denote by $\Gamma_H(\calP, \frakg / \frakh)$ the space of $H$-equivariant maps $\calP \to \frakg /\frakh$.

\begin{proposition}
\label{prop-map-r}
The map $r : \Gamma_H(\calP, \frakg) \to \Gamma_H(\calP, \frakg/\frakh)$, defined by $r(v)(p)\defeq [v(p)]$ for any $v \in \Gamma_H(\calP, \frakg)$ and $p \in \calP$, is such that the short exact sequence of $C^\infty(\calM)$-modules
\begin{equation}
\label{eq-sec-i-r}
\begin{tikzcd}[column sep=25pt, row sep=30pt]
\algzero
	\arrow[r]
& \Gamma_H(\calP, \frakh)
	\arrow[r, "i"]
& \Gamma_H(\calP, \frakg)
	\arrow[r, "r"]
& \Gamma_H(\calP, \frakg / \frakh)
	\arrow[r]
& \algzero
\end{tikzcd}
\end{equation}
is  exact. In particular, $r$ is surjective.
\end{proposition}

\begin{proof}
By definition, one has 
\begin{equation*}
r(v)(ph) = [v(ph)] = [\Ad_{h^{-1}} \circ v(p)] = \Ad_{h^{-1}} [v(p)] = \Ad_{h^{-1}} \circ r(v)(p),
\end{equation*}
so that $r(v) \in \Gamma_H(\calP, \frakg/\frakh)$.  Clearly this is a morphism of $C^\infty(\calM)$-modules and one has $\ker r = \Gamma_H(\calP, \frakh)$. 

Let  $\{ f_i\}$ be a partition of unity associated to a good covering $\{ \calU_i\}$ of $\calM$ such that $\calP_{| \calU_i}$ is trivial. For any $u \in \Gamma_H(\calP, \frakg/\frakh)$, let $u_i \defeq f_i u \in \Gamma_H(\calP_{| \calU_i}, \frakg/\frakh)$. Using local sections $s_i : \calU_i \to \calP_{| \calU_i}$, let $u_{\text{loc},i} \defeq s_i^* u_i : \calU_i \to \frakg/\frakh$. Then, for any $i$, since $\calU_i$ is connected and simply connected, we can find a map $v_{\text{loc},i} : \calU_i \to \frakg$ such that $[v_{\text{loc},i}(x)] = u_{\text{loc},i}(x)$ for any $x \in \calU_i$. Those $v_{\text{loc},i}$ define $H$-equivariant maps $v_i : \calP \to \frakg$ with supports in $\calP_{| \calU_i}$ such that $r(v_i) = u_i$. Then $v \defeq \sum_i v_i \in \Gamma_H(\calP, \frakg)$ satisfies $r(v) = r(\sum_i v_i) = \sum_i r(v_i) =  \sum_i f_i u = u$, which proves the surjectivity of $r$.  
\end{proof}

Prop.~\ref{prop-map-r} induces the short exact sequence of $C^\infty(\calM)$-modules
\begin{equation}
\label{eq-sec-j-rhi}
\begin{tikzcd}[column sep=25pt, row sep=30pt]
\algzero
	\arrow[r]
& \Gamma_H(\calP, \frakh)
	\arrow[r, "j"]
& \Gamma_G(\calQ, \frakg)
	\arrow[r, "\hr"]
& \Gamma_H(\calP, \frakg / \frakh)
	\arrow[r]
& \algzero
\end{tikzcd}
\end{equation}
where
\begin{equation}
\label{eq-def-hr}
\hr \defeq r \circ \hi^{-1}.
\end{equation}

\subsection{\texorpdfstring{The map $J$}{The map J}}

The aim of this subsection is to define a map $J : \Gamma_H(\calP) \to \Gamma_G(\calQ)$ which extends the map $j : \Gamma_H(\calP, \frakh) \to \Gamma_G(\calQ, \frakg)$. 

\begin{proposition}
\label{prop-def-J-properties}
For any $\frakX \in \Gamma_H(\calP)$, the vector field $\frakX \oplus 0 \in \Gamma(T\calP \oplus TG)$ satisfies \eqref{eq-tXxi-Q} and is $R^{\calP \times G}$-invariant, so that it defines $J(\frakX) \in \Gamma_G(\calQ)$.  The map
\begin{align*}
J : \Gamma_H(\calP) \to \Gamma_G(\calQ)
\end{align*}
is a morphism of Lie algebras and $C^\infty(\calM)$-modules such that
\begin{align}
\label{eq-J-iota-rho}
J \circ \iota_\calP &= \iota_\calQ \circ j,
&
\rho_\calQ \circ J &= \rho_P.
\end{align}
For any $v \in \Gamma_H(\calP, \frakh)$, $\hv \in \Gamma_H(\calP, \frakg)$, and $\frakX \in \Gamma_H(\calP)$, one has 
\begin{align}
\label{eq-j-fX-J-j}
j(\frakX \cdot v) &= J(\frakX) \cdot j(v),
&
\hi(\frakX \cdot \hv) &= J(\frakX) \cdot \hi(\hv).
\end{align}
and $J(\frakX)_{| [p,e]} = T_p \zeta (\frakX_{| p})$ for any $p \in \calP$.
\end{proposition}

\begin{proof}
Let $\calP \times \bbR \ni (p, t) \mapsto \phi_\frakX(p, t)$ be the flow of $\frakX$, which satisfies the right equivariance property $\phi_\frakX(p h, t) = \phi_\frakX(p, t)h$ for any $h \in H$ \cite[Sect.~4.2]{LazzMass12a}. The flow of $\frakX \oplus 0$ is then $\calP \times G \times \bbR \ni (p,g, t) \mapsto (\phi_\frakX(p, t), g) \in \calP \times G$. Then
\begin{align*}
T_{(ph, h^{-1} g)} \Pi_\calQ (\frakX_{| p h} \oplus 0_{| h^{-1} g})
&= \frac{d}{dt} \Pi_\calQ( \phi_\frakX(p h, t), h^{-1} g )_{| t=0}
= \frac{d}{dt} \Pi_\calQ( \phi_\frakX(p, t), g )_{| t=0}
= T_{(p, g)} \Pi_\calQ (\frakX_{| p} \oplus 0_{| g}),
\end{align*}
which proves \eqref{eq-tXxi-Q}. For any $g' \in G$, one has
\begin{align*}
T_{(p, g)} R^{\calP \times G}_{g'} (\frakX_{| p} \oplus 0_{| g})
&= \frac{d}{dt} R^{\calP \times G}_{g'}\left( \phi_\frakX(p, t), g \right)_{| t=0}
= \frac{d}{dt} \left( \phi_\frakX(p, t), g g' \right)_{| t=0}
= \frakX_{| p} \oplus 0_{| g g'},
\end{align*}
which proves the $R^{\calP \times G}$-invariance of $\frakX \oplus 0$. Notice that the flow of $J(\frakX)$ is $([p,g], t) \mapsto [\phi_\frakX(p, t), g]$  and this proves $J(\frakX)_{| [p,e]} = T_p \zeta (\frakX_{| p})$.

Since $\Gamma_H(\calP) \to \Gamma(T\calP \oplus TG)$, $\frakX \mapsto \frakX \oplus 0$ is a morphism of Lie algebras and $C^\infty(\calM)$-modules, so is $J$.

The flow of $\iota_\calP(v)$ for $v \in \Gamma_H(\calP, \frakh)$ is $(p,t) \mapsto p e^{- t v(p)}$ so that the flow of $\iota_\calP(v) \oplus 0$ is $(p, g,t) \mapsto (p e^{- t v(p)}, g)$. Then the flow of $J \circ \iota_\calP(v)$ is $([p,g], t) \mapsto [p e^{- t v(p)}, g]$. On the other hand, the flow of $\iota_\calQ \circ j(v)$ is $([p,g], t) \mapsto [p,g] e^{- t j(v)([p, g])}$ and $[p,g] e^{- t j(v)([p, g])} = [p,g] e^{- t \Ad_{g^{-1}} \circ i \circ v(p)} = [p, e^{- t i \circ v(p)} g] = [p e^{- t v(p)}, g]$ since $e^{- t i \circ v(p)} \in H \subset G$. So, the flows of $J \circ \iota_\calP(v)$ and $\iota_\calQ \circ j(v)$ coincide, which proves $J \circ \iota_\calP = \iota_\calQ \circ j$. By definition, one has $\rho_\calQ \circ J(\frakX) = (\rho_\calP \oplus 0)(\frakX \oplus 0) = \rho_\calP(\frakX)$, which proves $\rho_\calQ \circ J = \rho_\calP$.

We compute $J(\frakX) \cdot j(v)$ at the level of $\calP \times G$, where we lift $j(v)$ as $\hj(v)$ and $J(\frakX)$ as $\frakX \oplus 0$. Since this vector field acts only on the variables along $\calP$, one gets $((\frakX \oplus 0)\cdot \hj(v))(p,g) = \Ad_{g^{-1}} ((\frakX \cdot v)(p)) = \hj(\frakX \cdot v)(p, g)$, so that, projecting back on $\calQ$, one has $J(\frakX) \cdot j(v) = j(\frakX \cdot v)$.

Using the same method, lifting $\hi(\hv)$ as $(p, g) \mapsto \Ad_{g^{-1}} \circ \hv(p)$, we get $\hi(\frakX \cdot \hv) = J(\frakX) \cdot \hi(\hv)$.
\end{proof}

Using results of Sect.~\ref{subsec-transport-connection-P-Q}, we can give another defining expression for $J$.

\begin{proposition}
\label{prop-def-alt-J}
Let $\omega^\calP$ be a connection $1$-form on $\calP$. Then, for any $\frakX \in \Gamma_H(\calP)$, let $v \in \Gamma_H(\calP, \frakh)$ be such that $\frakX = \nabla^\calP_X + \iota_\calP(v)$ with $X = \rho_\calP(\frakX)$. Then $J : \Gamma_H(\calP) \to \Gamma_G(\calQ)$ is given by
\begin{equation*}
J(\frakX) = \nabla^\calQ_X + \iota_\calQ \circ j (v),
\end{equation*}
where the RHS is independent of the choice of the connection $\omega^\calP$. In particular, one has
\begin{equation}
\label{eq-J-nabla-P-nabla-Q}
J \circ \nabla^\calP = \nabla^\calQ.
\end{equation}
\end{proposition}

\begin{proof}
From Lemmas~\ref{lem-iotaQj-iotaP} and \ref{lem-nablaQ-TpiQ-nablaP}, we know that $\nabla^\calQ_X + \iota_\calQ \circ j (v) \in \Gamma_G(\calQ)$ comes from $\big(\nabla^\calP_X + \iota_\calP(v)\big) \oplus 0 = \frakX \oplus 0 \in \Gamma(T\calP \oplus TG)$: this gives the proposed expression for $J(\frakX)$.

Let $\omega'^\calP = \omega^\calP + \connexionShift$ be another connection on $\calP$ and $\frakX \in \Gamma_H(\calP)$. Then, one has $\frakX = \nabla^\calP_X + \iota_\calP(v) = \nabla'^\calP_X + \iota_\calP(v')$ for $v, v' \in \Gamma_H(\calP, \frakh)$. From \eqref{eq-nablaprimeP-nablaP-vXshift}, one gets $v = v' + v^{X, \connexionShift}$. Then, using \eqref{eq-nablaprimeQ-nablaQ-vXshift}, $\nabla'^\calQ_X + \iota_\calQ \circ j (v') = \nabla^\calQ_X + \iota_\calQ \circ j(v^{X, \connexionShift}) + \iota_\calQ \circ j (v') = \nabla^\calQ_X + \iota_\calQ \circ j (v)$, which shows that the RHS does not depend on the choice of $\omega^\calP$. \eqref{eq-J-nabla-P-nabla-Q} is proved for $\frakX = \nabla^\calP_X$.
\end{proof}

This relation clearly shows how $J$ extends $j$: the map $j$ applies to the “inner” part $v$ of $\frakX$, while we need the associated connection $\nabla^\calQ$ to take care of the geometric part $X$.

\subsection{\texorpdfstring{The map $R$}{The map R}}

For any $\hfX, \hfY \in \Gamma_G(\calQ)$, we define the following equivalence relation: $\hfX \sim \hfY$ if and only if there exists $\frakZ \in \Gamma_H(\calP)$ such that $\hfX - \hfY = J(\frakZ)$. We denote by $R : \Gamma_G(\calQ) \to \Gamma_G(\calQ)/{\sim}$ the quotient map for this equivalence relation. In the following, we use Prop.~\ref{prop-def-alt-J}.

Let $\hfX = \nabla^\calQ_X + \iota_\calQ(v)$, $\hfY = \nabla^\calQ_Y + \iota_\calQ(w)$ and $\frakZ = \nabla^\calP_Z + \iota_\calP(z)$, with $v,w \in \Gamma_G(\calQ, \frakg)$ and $z \in \Gamma_H(\calP, \frakh)$. Then we have $\nabla^\calQ_X - \nabla^\calQ_Y + \iota_\calQ(v-w) = \nabla^\calQ_Z + \iota_\calQ \circ j(z)$. Applying $\rho_\calQ$ on this relation, we get $X-Y=Z$, so that the relation reduces to $v-w = j(z)$. With $\tv = \hi^{-1}(v), \tw = \hi^{-1}(w) \in \Gamma_H(\calP, \frakg)$, the equivalence relation $\hfX \sim \hfY$ reduces to the equivalence relation $\tv \sim \tw$ in $\Gamma_H(\calP, \frakg)$ defining the map $r : \Gamma_H(\calP, \frakg) \to \Gamma_H(\calP, \frakg/\frakh)$ of Prop.~\ref{prop-map-r}.

This shows that 
\begin{equation}
\label{eq-R-r}
R(\hfX) = R( \nabla^\calQ_X + \iota(\tv)) = r(\tv),
\end{equation}
from which we deduce 
\begin{equation*}
\Gamma_G(\calQ)/{\sim}  \simeq \Gamma_H(\calP, \frakg/\frakh).
\end{equation*}
It is straightforward to check that $R$ does not depend on the choice of the connection $\omega^\calP$.

From these results, we get the short exact sequence of $C^\infty(\calM)$-modules
\begin{equation}
\label{eq-sec-J-R}
\begin{tikzcd}[column sep=25pt, row sep=30pt]
\algzero
	\arrow[r]
& \Gamma_H(\calP)
	\arrow[r, "J"]
& \Gamma_G(\calQ)
	\arrow[r, "R"]
& \Gamma_H(\calP, \frakg / \frakh)
	\arrow[r]
& \algzero
\end{tikzcd}
\end{equation}

\subsection{The exact commutative diagram}
\label{sec exact comm diagram}

Collecting all the maps defined so far, we get the following diagram:
\begin{equation}
\label{eq-diagram}
\begin{tikzcd}[column sep=30pt, row sep=35pt, nodes={inner sep=2pt}]
&[0pt] %
&[-55pt] 0 
	\arrow[d] 
&[-55pt] %
& 0 
	\arrow[d]  
& 0 
	\arrow[d] 
&[-10pt] %
\\[-10pt]
0 
	\arrow[rr]
& %
&|[alias=GammaPh]| \Gamma_H(\calP, \frakh) 
	\arrow[rr, "\iota_\calP"] 
	\arrow[dl, "i"'] 
& %
& \Gamma_H(\calP) 
	\arrow[r, "\rho_\calP"] 
	\arrow[dd, "J" ]
& \Gamma(T\calM) 
	\arrow[r] 
	\arrow[dd, equal]
& 0
\\
& \Gamma_H(\calP, \frakg) 
	\arrow[dddr, pos=0.70, "r"'] 
	\arrow[drrr, start anchor=east, pos=0.75, end anchor={[yshift=+0.5ex]west}, "\iota"] 
	\arrow[ddrr, outer sep = -0.2ex, "\hi", "\simeq"'] 
& %
& %
& %
& %
& %
\\[-30pt]
|[alias=zeromiddle]| 0 
	\arrow[ur, end anchor=west] 
& %
& %
& %
& \Gamma_G(\calQ) 
	\arrow[r, "\rho_\calQ"] 
	\arrow[dd,  "R"]
& \Gamma(T\calM) 
	\arrow[r] 
	\arrow[dd]
& 0
\\[-30pt]
& %
& %
&|[alias=GammaQg]| \Gamma_G(\calQ, \frakg) 
	\arrow[ru, start anchor=east, , end anchor={[yshift=-0.5ex]west}, "\iota_\calQ"'] 
	\arrow[dl, pos=0.35, "\hr"]
& %
& %
& %
\\
0 \arrow[rr]
& %
& \Gamma_H(\calP, \frakg / \frakh) 
	\arrow[rr, equal] \arrow[d] 
& %
& \Gamma_H(\calP, \frakg / \frakh) 
	\arrow[r] \arrow[d]
& 0
& %
\\[-10pt]
& %
& 0 
& %
& 0 
&
&
	\arrow[from=zeromiddle, to=GammaQg, crossing over, crossing over clearance=1ex, end anchor=west]
	\arrow[from=GammaPh, to=GammaQg, crossing over, crossing over clearance=1ex, pos=0.25, "j"]
\end{tikzcd}  
\end{equation}

\begin{proposition}
The diagram of $C^\infty(\calM)$-modules \eqref{eq-diagram} is exact and commutative.
\end{proposition}

\begin{proof}
The exactness of this diagram is a direct consequence of the exactness of the short exact sequences \eqref{eq-sec-Atiyah-P}, \eqref{eq-sec-Atiyah-Q}, \eqref{eq-sec-i-r}, \eqref{eq-sec-j-rhi}, and \eqref{eq-sec-J-R}. The commutativity comes from \eqref{eq-j-i-hi}, \eqref{eq-def-iota}, \eqref{eq-J-iota-rho}, and \eqref{eq-R-r}.
\end{proof}

Notice that in \eqref{eq-diagram} the two top horizontal lines are short exact sequences of Lie algebras.

\smallskip
\begin{remark}
\label{rmk Mackenzie}
All the maps described in this diagram are constructed using sections of vector bundles, some of them being identified as vector valued equivariant maps on principal fiber bundles. As mentioned in the introduction, the first two rows of this diagram have been described in the language of fiber bundles, see diagram~(5) p.~289 in \cite{Mack87a}, in a more general situation of a bundle map $F(\Id, \phi) :  \calP(\calM, H) \to  \calP'(\calM, H')$ where $\Id : \calM \to \calM$ is the identity map, $\phi : H \to H'$ is a Lie groups morphism and $F :  \calP \to  \calP'$ satisfies $F(ph) = F(p) \phi(h)$ for any $(p,h) \in  \calP \times H$. It is straightforward to check that our maps $j$ and $J$ are related respectively to the maps $F^+_*$ and $F_*$ in diagram~(5), where here $F = \zeta$ is the bundle map inclusion $\calP \to \calQ$ and $\phi : H \hookrightarrow G$ is the inclusion. The new results in \eqref{eq-diagram} are then the identification of the last row to get vertical short exact sequences and the (direct) rewriting of all the maps in the language of sections adapted to physical applications.
\end{remark}

\subsection{The diagram in a local trivialization}
\label{sec The diagram in a local trivialization}

Let us use the notations of Sect.~\ref{sec-local-trivializations}. We have $\Gamma_H(\calP, \frakh) \locto \Gamma(\calU \times \frakh)$, $\Gamma_H(\calP, \frakg) \locto \Gamma(\calU \times \frakg)$, $\Gamma_H(\calP, \frakg/\frakh) \locto \Gamma(\calU \times \frakg/\frakh)$, and $\Gamma_G(\calQ, \frakg) \locto \Gamma(\calU \times \frakg)$ by the associations $v \locto \gamma \defeq s^*v$ for $v \in \Gamma_H(\calP, \frakh)$, $\hv \locto \hgamma \defeq s^*\hv$ for $\hv \in \Gamma_H(\calP, \frakg)$, $\bv \locto \bgamma \defeq s^*\bv$ for $\bv \in \Gamma_H(\calP, \frakg/\frakh)$, and $\hv \locto \hgamma \defeq \hs^* \hv$ for $\hv \in \Gamma_G(\calQ, \frakg)$.

Let us use the subscript $\loc$ to designate the local versions of the maps in \eqref{eq-diagram}. Then it is straightforward to establish that
\begin{align*}
&\rho_{\calP, \loc}(X \oplus \gamma) = X,
&
&\rho_{\calQ, \loc}(X \oplus \hgamma) = X,
\\
&\iota_{\calP, \loc}(\gamma) = 0 \oplus \gamma,
&
&\iota_{\calQ, \loc}(\hgamma) = 0 \oplus \hgamma,
\\
& i_\loc(\gamma) = i_\frakh (\gamma),
&
& j_\loc(\gamma) = i_\frakh (\gamma),
\\
&\iota_\loc(\hgamma) = \hi(\hgamma) = 0 \oplus \hgamma,
&
&r_\loc(\hgamma) =\hr_\loc(\hgamma) = [\hgamma],
\\
&J_\loc(X \oplus \gamma) = X \oplus i_\frakh (\gamma),
&
&R_\loc(X \oplus \hgamma) = [\hgamma].
\end{align*}

Notice that all the maps in \eqref{eq-diagram} reduce to very simple expressions in a local trivialization. This confirms the naturalness of the construction of this diagram, which will be further reinforced in the forthcoming considerations about connections.

\section{Connections and metrics}  
\label{sec Connections and metrics}

In \cite{LazzMass12a, FourLazzMass13a}, connections and generalized connections were considered in the framework of transitive Lie algebroids, in particular on Atiyah Lie algebroids. An Ehresmann connection on $\calP$ is in one-to-one correspondence with a splitting of \eqref{eq-sec-Atiyah-P}, $\nabla^\calP : \Gamma(T\calM) \to \Gamma_H(\calP)$, to which we can associate a unique map $\alpha^\calP : \Gamma_H(\calP) \to \Gamma_H(\calP, \frakh)$. In \cite{LazzMass12a}, this map was considered as a $1$-form in the differential calculus $\Omega^\grast_\lie(\calP, \frakh)$ described in Sec.~\ref{sec-differential-calculi}. In the following, we first characterized Ehresmann connections, and then Cartan connections, within the diagram~\eqref{eq-diagram}.

\subsection{Ehresmann connections}
\label{sec-Ehresmann-connect}

Let $\nabla^\calP : \Gamma(T\calM) \to \Gamma_H(\calP)$ be an Ehresmann connection on $\calP$ and $\alpha^\calP : \Gamma_H(\calP) \to \Gamma_H(\calP, \frakh)$ its associated unique map defined by $\frakX = \nabla^\calP_X - \iota_\calP \circ \alpha^\calP(\frakX)$ for any $\frakX \in \Gamma_H(\calP)$ where $X \defeq \rho_\calP(\frakX)$, see Section~\ref{subsec-transport-connection-P-Q}. 
Let $\nabla^\calQ = J \circ \nabla^\calP : \Gamma(T\calM) \to \Gamma_G(\calQ)$ be the associated connection on $\calQ$ (see \eqref{eq-J-nabla-P-nabla-Q}) and $\alpha^\calQ : \Gamma_G(\calQ) \to \Gamma_G(\calQ, \frakh)$ such that $\hfX = \nabla^\calQ_X - \iota_\calQ \circ \alpha^\calQ(\hfX)$ for any $\hfX \in \Gamma_G(\calQ)$.

\begin{proposition}
\label{prop-alphaQ-J-j-alphaP}
One has $\alpha^\calQ \circ J =j \circ \alpha^\calP$ and $R = -\hr \circ \alpha^\calQ$. 
\end{proposition}

\begin{proof}
For any $\frakX \in \Gamma_H(\calP)$ written as $\frakX = \nabla^\calP_X - \iota_\calP \circ \alpha^\calP(\frakX)$, one has, from Prop.~\ref{prop-def-alt-J}, $J(\frakX) = \nabla^\calQ_X - \iota_\calQ \circ j \circ \alpha^\calP(\frakX)$. On the other hand, since $J(\frakX) \in \Gamma_G(\calQ)$, one has $J(\frakX) = \nabla^\calQ_X - \iota_\calQ \circ \alpha^\calQ \circ J(\frakX)$. By injectivity of $\iota_\calQ$, we get $j \circ \alpha^\calP(\frakX) = \alpha^\calQ \circ J(\frakX)$.

For any $\hfX \in \Gamma_G(\calQ)$, using \eqref{eq-R-r} and \eqref{eq-def-hr}, one has $R(\hfX) = R(\nabla^\calQ_X - \iota_\calQ \circ \alpha^\calQ(\hfX)) = R(\nabla^\calQ_X - \iota \circ \hi^{-1} \circ \alpha^\calQ(\hfX)) =  - \hr \circ \alpha^\calQ(\hfX)$.
\end{proof}

We denote by $\omegaLie^\calP \in \Omega^1_\lie(\calP, \frakh)$ the associated $1$-form $\omegaLie^\calP = -\alpha^\calP$, so that $\omegaLie^\calP \circ \iota_\calP = \Id$. In the same way, define $\omegaLie^\calQ = - \alpha^\calQ \in  \Omega^1_\lie(\calQ, \frakg)$, which satisfies 
\begin{align}
\label{eq omegaLieQ J = j omegaLieP and R = hr omegaLieQ}
\omegaLie^\calQ \circ J &= j \circ \omegaLie^\calP,
&
R &= \hr \circ \omegaLie^\calQ
\end{align}
by Prop.~\ref{prop-alphaQ-J-j-alphaP}. Then, for any $\frakX \in \Gamma_H(\calP)$ and $\hfX \in \Gamma_G(\calQ)$, we have
\begin{align*}
\frakX &= \nabla^\calP_X + \iota_\calP \circ \omegaLie^\calP(\frakX),
&
\hfX &= \nabla^\calQ_X + \iota_\calQ \circ \omegaLie^\calQ(\hfX).
\end{align*}
Let $\omega^\calP \in \Omega^1(\calP) \otimes \frakh$ be the connection $1$-form associated to $\nabla^\calP$. Then $\omegaLie^\calP$ can also be directly defined by 
\begin{equation}
\label{eq-def-omegaLie}
\omegaLie^\calP(\frakX)(p) \defeq - \omega^\calP_{| p}(\frakX_{| p}) \text{ for any $\frakX \in \Gamma_H(\calP)$ and $p \in \calP$}. 
\end{equation}
Notice that it is a departure from our convention in \cite{LazzMass12a, FourLazzMass13a} where there the present $\alpha^\calP$ plays the role of the associated $1$-form to $\nabla^\calP$. This new convention is used here to be compatible with a similar construction for Cartan connections (Sect.~\ref{sec Cartan connections} and in particular the chosen normalization to get Prop.~\ref{prop commutativity varpiLie}). 

Then Prop.~\ref{prop-alphaQ-J-j-alphaP} implies the commutativity of the maps with solid lines in the following diagram.
\begin{equation}
\label{eq-diagram-ehresmann}
\begin{tikzcd}[column sep=30pt, row sep=30pt]
\Gamma_H(\calP, \frakh) 
	\arrow[r, "\iota_\calP", dashed] 
	\arrow[d, "j"] 
& \Gamma_H(\calP) 
	\arrow[l, bend right=25, "\omegaLie^\calP"'] 
	\arrow[r, "\rho_\calP", dashed] 
	\arrow[d, "J" ]
& \Gamma(T\calM) 
	\arrow[l, bend right=25, "\nabla^\calP"'] 
	\arrow[d, equal]
\\
\Gamma_G(\calQ, \frakg) 
	\arrow[r, "\iota_\calQ", dashed] 
	\arrow[d,  "\hr"]
& \Gamma_G(\calQ) 
	\arrow[l, bend right=25, "\omegaLie^\calQ"'] 
	\arrow[r, "\rho_\calQ", dashed] 
	\arrow[d,  "R"]
& \Gamma(T\calM) 
	\arrow[l, bend right=25, "\nabla^\calQ"']  
\\
\Gamma_H(\calP, \frakg / \frakh) 
	\arrow[r, equal]
& \Gamma_H(\calP, \frakg / \frakh)
&
\end{tikzcd}  
\end{equation}

Let $\omega^\calP \in \Omega^1(\calP) \otimes \frakh$ be the connection $1$-form defining $\nabla^\calP$. Using notations from Sect.~\ref{sec-local-trivializations} and \ref{subsec-transport-connection-P-Q}, one has $\nabla^\calP_X = S(X \oplus A^\calP(X))$ for any $X \in \Gamma(T\calU)$. For any $\frakX \in \Gamma_H(\calP_{| \calU})$ with $\frakX \locto X \oplus \gamma$, one has, on the one hand, using \eqref{eq-S-notnabla-Psi}, $\frakX = S(X \oplus \gamma) = \mrnabla_X + \iota_\calP \circ \Psi_\frakh(\gamma)$ and on the other hand, $\frakX = \nabla^\calP_X + \iota_\calP \circ \omegaLie^\calP(\frakX)$. Since $\nabla^\calP_X - \mrnabla_X = S(X \oplus A^\calP(X)) - S(X \oplus 0) = S(0 \oplus A^\calP(X)) = \iota_\calP \circ \Psi_\frakh \circ A^\calP(X)$, one has $\iota_\calP \circ \omegaLie^\calP \circ S(X \oplus \gamma) = - \nabla^\calP_X + \mrnabla_X + \iota_\calP \circ \Psi_\frakh(\gamma) = - \iota_\calP \circ \Psi_\frakh( A^\calP(X) - \gamma)$. Using the same computation on $\calQ$, one finally gets
\begin{align}
\label{eq-local-omegalie}
\omegaLie^\calP &\locto - A^\calP + \theta_\frakh,
&
\omegaLie^\calQ &\locto - A^\calQ + \theta_\frakg,
\end{align}
where $\theta_\frakh$ and $\theta_\frakg$ are the Maurer-Cartan $1$-forms on $H$ and $G$ respectively.

\smallskip
Let $\frakX = \nabla^\calP_X + \iota_\calP(v)$ and $\frakY = \nabla^\calP_Y + \iota_\calP(w)$ in $\Gamma_H(\calP)$ with $X, Y \in \Gamma(T\calM)$ and $v, w \in \Gamma_H(\calP, \frakh)$. Then, let $\Omega^\calP(X,Y) \in \Gamma_H(\calP, \frakh)$ be the curvature $2$-form of $\nabla^\calP$ \cite[Prop.~3.9]{LazzMass12a} defined by $\iota_\calP(\Omega^\calP(X,Y)) \defeq [\nabla^\calP_X, \nabla^\calP_Y] - \nabla^\calP_{[X,Y]}$, and let $\Omega^\calP \locto \Omega^\calP_\loc$ be its local curvature $2$-form. Then one has $[\frakX, \frakY] = [\nabla^\calP_X, \nabla^\calP_Y] + [\nabla^\calP_X, \iota_\calP(w)] - [\nabla^\calP_Y, \iota_\calP(v)] + [\iota_\calP(v), \iota_\calP(w)] = \nabla^\calP_{[X,Y]} + \iota_\calP( \Omega^\calP(X,Y) + \nabla^\calP_X \cdot w - \nabla^\calP_Y \cdot v + [v,w])$, so that, using \eqref{eq-trivial-X-v} with $\frakX \locto X \oplus (A^\calP(X) + \gamma)$ and $\frakY \locto Y \oplus (A^\calP(Y) + \eta)$ (\textit{i.e.} $v\locto \gamma$ and $w\locto \eta$), one has
\begin{align}
\label{eq local bracket}
[\frakX, \frakY] \locto [X,Y] \oplus (A^\calP([X,Y]) + \Omega^\calP_\loc(X,Y) + X \cdot \eta + [ A^\calP(X), \eta] - Y \cdot \gamma -[ A^\calP(Y), \gamma] + [\gamma, \eta]),
\end{align}
or equivalently $\Omega^\calP(X,Y) + \nabla^\calP_X \cdot w - \nabla^\calP_Y \cdot v + [v,w] \locto \Omega^\calP_\loc(X,Y) + X \cdot \eta + [ A^\calP(X), \eta] - Y \cdot \gamma -[ A^\calP(Y), \gamma] + [\gamma, \eta]$.

\subsection{Cartan connections}
\label{sec Cartan connections}

From now on, we suppose that $\calP$ defines a Cartan geometry for the groups $H \subset G$.

Let $\varpi \in \Omega^1(\calP) \otimes \frakg$ be a Cartan connection on $\calP$. This $1$-form satisfies:
\begin{align*}
&{R^\calP_h}^*\varpi = \Ad_{h^{-1}}\varpi \text{ for any $h \in H$},
\\                                             
&\varpi(\xi^v) = \xi \text{ for any $\xi \in \frakh$},
\\
&\varpi_{| p} : T_p\calP \to \frakg \text{ is an isomorphism}.   
\end{align*}
We associate to $\varpi$ the map $\varpiLie : \Gamma_H(\calP) \to \Gamma_H(\calP, \frakg)$ defined by (as a generalization of \eqref{eq-def-omegaLie})
\begin{equation}
\label{eq-def-varpiLie}
\varpiLie(\frakX)(p) \defeq - \varpi_{| p}(\frakX_{| p}) \text{ for any $\frakX \in \Gamma_H(\calP)$ and $p \in \calP$}. 
\end{equation}
The map $p \mapsto \varpiLie(\frakX)(p) \in \frakg$ is indeed $H$-equivariant: for any $p \in \calP$ and $h \in H$,
\begin{align*}
\varpiLie(\frakX)(ph) 
&= - \varpi_{| ph}(\frakX_{| ph}) 
= - \varpi_{|ph}(T_p R^\calP_h (\frakX_{| p}) ) 
= - ({R^\calP_h}^*\varpi)_{| p}(\frakX_{| p}) 
=  - \Ad_{h^{-1}} \varpi_{| p}(\frakX_{| p}) 
= \Ad_{h^{-1}} \circ \varpiLie(\frakX)(p).
\end{align*}

\begin{proposition}
\label{prop commutativity varpiLie}
$\varpiLie : \Gamma_H(\calP) \to \Gamma_H(\calP, \frakg)$ is an isomorphism of $C^\infty(\calM)$-modules and the following diagram
\begin{equation}
\label{eq-diagram-triangle-upleft}
\begin{tikzcd}[column sep=30pt, row sep=30pt]
 \Gamma_H(\calP, \frakh) \arrow[r, "\iota_\calP"] \arrow[d, "i"] 
& \Gamma_H(\calP) \arrow[dl,  "\varpiLie",  "\simeq"']
\\
 \Gamma_H(\calP, \frakg)
 &
\end{tikzcd} 
\end{equation}
is commutative: $\varpiLie \circ \iota_\calP = i$.
\end{proposition}

\begin{proof}
Using the isomorphism $\varpi_{| p}$, to any $v \in \Gamma_H(\calP, \frakg)$ we associate $\frakX \in \Gamma(T\calP)$ by $\frakX_{| p} \defeq -\varpi_{| p}^{-1}(v(p))$. Let us check that $\frakX \in \Gamma_H(\calP)$: one has $\varpi_{| ph}(T_p R^\calP_h (\frakX_{| p})) = ({R^\calP_h}^* \varpi_{| p})(\frakX_{| p}) = \Ad_{h^{-1}} \varpi_{| p}(\frakX_{| p}) = - \Ad_{h^{-1}}v(p) = - v(ph) = \varpi_{| ph}(\frakX_{| ph})$ so that $\frakX_{| ph} = T_p R^\calP_h (\frakX_{| p})$. By construction $\varpiLie(\frakX) = v$, which proves the surjectivity.

Suppose that $\frakX \in \Gamma_H(\calP)$ is such that $\varpiLie(\frakX)=0$. Then, for any $p \in \calP$, $\varpi_{| p}(\frakX_{| p}) = 0$, so that $\frakX_{| p} = 0$ since $\varpi_{| p}$ is injective. This proves the injectivity.

For any $v\in \Gamma_H(\calP, \frakh)$, one has $\iota_\calP(v)_{| p} = - [v(p) ]^v_{| p}$ so that $\varpiLie(\iota_\calP(v))(p) = - \varpi_{| p}(\iota_\calP(v)_{| p}) = \varpi_{| p}\left( [v(p) ]^v_{| p} \right) = v(p)$. Since $i$ is just the inclusion map, we get $\varpiLie \circ \iota_\calP = i$.
\end{proof}

Since $r \circ \varpiLie \circ \iota_\calP = r \circ i = 0$, $\iota_\calP(\Gamma_H(\calP, \frakh))$ is in the kernel of the map $r\circ \varpiLie$. This induces a map $\tvarpiLie : \Gamma(T\calM) \to \Gamma_H(\calP, \frakg/\frakh)$ defined as follows: for any $X \in \Gamma(T\calM)$, consider any lift $\frakX \in \Gamma_H(\calP)$ such that $\rho_\calP(\frakX) = X$ (defined up to an element in $\iota_\calP(\Gamma_H(\calP, \frakh))$), then 
\begin{equation}
\label{eq-def-tvarpiLie}
\tvarpiLie(X) \defeq r \circ \varpiLie(\frakX).
\end{equation}

\begin{proposition}  \label{prop-varpiLie-tilde}
$\tvarpiLie : \Gamma(T\calM) \to \Gamma_H(\calP, \frakg/\frakh)$ is an isomorphism of $C^\infty(\calM)$-modules.
\end{proposition}

\begin{proof}
To show that $\tvarpiLie$ is an isomorphism, apply the five lemma to the commutative diagram
\begin{equation} \label{eq-diag-five}
\begin{tikzcd}[column sep=30pt, row sep=25pt]
0 \arrow[r] \arrow[d, equal] 
& \Gamma_H(\calP, \frakh) \arrow[r, "\iota_\calP"] \arrow[d, equal] 
& \Gamma_H(\calP) \arrow[r, "\rho_\calP"] \arrow[d,  "\varpiLie",  "\simeq"'] 
& \Gamma(T\calM) \arrow[r] \arrow[d, "\tvarpiLie",  "\simeq"'] 
& 0 \arrow[d, equal] 
\\
 0 \arrow[r] 
 & \Gamma_H(\calP, \frakh) \arrow[r, "i"] 
 & \Gamma_H(\calP, \frakg) \arrow[r, "r"] 
 & \Gamma_H(\calP, \frakg/\frakh)\arrow[r] 
 & 0
\end{tikzcd}
 \end{equation}
 \end{proof} 

We have thus seen that a Cartan connection induces maps in the diagram~\eqref{eq-diagram} so that we get the following maps of $C^\infty(\calM)$-modules (this diagram is not commutative everywhere):
\begin{equation}
\label{eq-diagram-cartan}
\begin{tikzcd}[column sep=30pt, row sep=30pt]
\Gamma_H(\calP, \frakh) \arrow[r, "\iota_\calP"] \arrow[d, "i"] 
& \Gamma_H(\calP) \arrow[r, "\rho_\calP"] \arrow[d, "J" ] \arrow[dl, "\varpiLie",  "\simeq"']
& \Gamma(T\calM) \arrow[d, equal] 
\\
\Gamma_H(\calP, \frakg) \arrow[r, "\iota"] \arrow[d,  "r"]
& \Gamma_G(\calQ) \arrow[r, "\rho_\calQ"] \arrow[d,  "R"]
& \Gamma(T\calM)  \arrow[dl, "\tvarpiLie",  "\simeq"'] 
\\
\Gamma_H(\calP, \frakg / \frakh) \arrow[r, equal] 
& \Gamma_H(\calP, \frakg / \frakh) 
\end{tikzcd}  
\end{equation}

Let us now show the opposite: that $\varpiLie$ as in \eqref{eq-diagram-cartan}, such that \eqref{eq-diagram-triangle-upleft} commutes, defines a Cartan connection $\varpi$ on $\calP$. Notice that $\varpi$ must be defined on all vector fields on $\calP$ while $\varpiLie$ is only defined on right-invariant vector fields. The trick is to use the following result. 

From \cite[Prop.~4.4]{LazzMass12a}, we know that for every open subset $\calU \subset \calM$ which is an open set for a chart on $\calM$ and which trivializes $\calP$, there exists a family of generators $\{ \frakX^i \}$ of right-invariant vector fields on $\calP_{| \calU}$ such that:
\begin{enumerate}[label=(\roman*)]
\item $\forall p \in \calP_{| \calU}$, $\{ \frakX^i_{| p} \}$ is a basis for $T_p\calP_{| \calU}$.

\item It generates $\Gamma_H(\calP_{| \calU})$ as a $C^\infty(\calU)$-module, and, for any $\frakX \in \Gamma_H(\calP_{| \calU})$, the decomposition $\frakX = f_i \frakX^i$, with $f_i \in C_H^\infty(\calP_{| \calU})\simeq C^\infty(\calU)$, is unique.

\item It generates $\Gamma(T \calP_{|\calU})$ as a $C^\infty(\calP_{|\calU})$-module, and, for any $\tX \in \Gamma(T \calP_{|\calU})$,  the decomposition $\tX = g_i \frakX^i$, with $g_i \in C^\infty(\calP_{|\calU})$, is unique. 

\item Two such families $\{ \frakX^i \}$ are related by linear combinations whose coefficients are in $C_H^\infty(\calP_{| \calU})$.
\end{enumerate}

We are going to define the Cartan connection $\varpi$ on $\calP_{|\calU}$ for any open subset as before with its associated family of generators $\{ \frakX^i \}$. 

Decompose any $\tX \in \Gamma(T \calP_{|\calU})$ as $\tX = g_i \frakX^i$ with $g_i \in C^\infty(\calP_{|\calU})$ and define $\varpi \in \Omega^1(\calP_{|\calU}) \otimes \frakg$ by $\varpi(\tX) \defeq - g_i \varpiLie(\frakX^i)$. This definition makes sense since the decomposition is unique and the RHS is independent of the choice of the family of generators: let $\frakY^j = {a^j}_i \frakX^i$ with ${a^j}_i \in C_H^\infty(\calP_{| \calU})$ be another family of generators, then $\tX = h_j \frakY^j = g_i \frakX^i$ with $h_j \in C^\infty(\calP_{|\calU})$, so that $h_j {a^j}_i = g_i$, which implies $g_i \varpiLie(\frakX^i) = h_j {a^j}_i \varpiLie(\frakX^i) = h_j \varpiLie({a^j}_i \frakX^i) = h_j \varpiLie(\frakY^j)$. It is important to notice that $\varpiLie$ is $C_H^\infty(\calP_{| \calU})$-linear (functions ${a^j}_i$) but not $C^\infty(\calP_{|\calU})$-linear (functions $g_i$ and $h_j$).

Using a partition of unity associated to a covering $\{ \calU_i\}$ of $\calM$ such that the $\calP_{| \calU_i}$ are trivial, this defines $\varpi \in \Omega^1(\calP) \otimes \frakg$.

\begin{proposition}
With the above definition of $\varpi$ in terms of $\varpiLie$, one has:
\begin{enumerate}[label=(\roman*)]
\item\label{item-varpi-varpilie} $\varpi(\frakX) = - \varpiLie(\frakX)$ for any $\frakX \in \Gamma_H(\calP)$. 

\item\label{item-varpixiv-xi} $\varpi(\xi^v)=\xi$ for any $\xi \in \frakh$.

\item\label{item-RPh-varpi} ${R^\calP_h}^* \varpi = \Ad_{h^{-1}} \varpi$ for any $h \in H$.

\item\label{item-varpip-iso} For any $p \in \calP$, $\varpi_{| p} : T_p \calP \to \frakg$ is an isomorphism.
\end{enumerate}
\end{proposition}

\begin{proof}
\ref{item-varpi-varpilie}: For $\frakX \in \Gamma_H(\calP_{|\calU})$, in the decomposition $\frakX = g_i \frakX^i$ the functions $g_i$ are in $C_H^\infty(\calP_{| \calU})$ so that $\varpi(\frakX) = - g_i \varpiLie(\frakX^i) = - \varpiLie(g_i \frakX^i) = - \varpiLie(\frakX)$ by $C_H^\infty(\calP_{| \calU})$-linearity of $\varpiLie$.

\ref{item-varpixiv-xi}: Let $\xi \in \frakh$. Fix a point $p \in \calP_{| \calU}$ and a smooth section $s : \calU \to \calP_{| \calU}$ such that $p = s(x_0)$. Define $v_0 \in \Gamma_H(\calP_{|\calU}, \frakh)$ by $v_0(s(x) h) \defeq \Ad_{h^{-1}} \xi$ for any $x \in \calU$ and $h \in H$, and let $f \in C^\infty(\calM)$ be such that $f = 1$ in a neighborhood of $x_0 \in \calU$ and $f = 0$ outside $\calU$. Then $v \defeq f v_0 \in \Gamma_H(\calP, \frakh)$ is such that $v(p) = v_0(s(x_0)) = \xi$. Then, using \ref{item-varpi-varpilie} and the commutativity of \eqref{eq-diagram-triangle-upleft}, $\varpi_{| p}( [\xi]^v_{| p}) = \varpi_{| p}( \iota_\calP(v)_{| p} ) = - \varpiLie(\iota_\calP(v))(p) = v(p) = \xi$.

\ref{item-RPh-varpi}: We compute this relation on $\calP_{| \calU}$. Let $\tX \in \Gamma(T \calP_{|\calU})$ with $\tX = g_i \frakX^i$. Then, for any $p \in \calP_{| \calU}$ and $h \in H$, 
\begin{align*}
\Ad_{h^{-1}} \circ \varpi_{| p} (\tX_{| p}) 
&= - \Ad_{h^{-1}} ( g_i(p) \varpiLie(\frakX^i)(p) ) 
= - g_i(p) \Ad_{h^{-1}}(\varpiLie(\frakX^i)(p))
= - g_i(p) \varpiLie(\frakX^i)(ph),
\\
({R^\calP_h}^* \varpi)_{| p} (\tX_{| p}) 
&= \varpi_{| ph} (T_pR^\calP_h (\tX_{| p}))
= \varpi_{| ph} (g_i(p) T_pR^\calP_h (\frakX^i_{| p}))
= \varpi_{| ph} (g_i(p) \frakX^i_{| ph})
= - g_i(p) \varpiLie(\frakX^i)(ph).
\end{align*}

\ref{item-varpip-iso}: The proof relies on a way to translate a “local” property on $\varpiLie$ over $\calP_{|\calU}$ to a “pointwise” property on $\varpi$ at $p \in \calP_{|\calU}$. Let us use the notations introduced in the proof of \ref{item-varpixiv-xi}.

Let $\xi \in \frakg$ and define $\hv_0 \in \Gamma_H(\calP_{|\calU}, \frakg)$ by $\hv_0(s(x)h) \defeq \Ad_{h^{-1}} \xi$ for any $x \in \calU$ and $h \in H$. Then $\hv \defeq f \hv_0 \in \Gamma_H(\calP, \frakg)$ is such that $\hv(p) = \xi$. Since $\varpiLie$ is an isomorphism, one defines $\frakX \defeq - \varpiLie^{-1}(\hv)$ and one gets $\varpi_{|p}(\frakX_{|p}) = - \varpiLie(\frakX)(p) = \xi$, so that $\varpi_{|p} : T_p \calP \to \frakg$ is surjective. Since $T_p\calP$ and $\frakg$ have same dimensions, $\varpi_{|p}$ is an isomorphism.
\end{proof}

We have then proved the following important equivalence of structures.
\begin{theorem}
\label{thm eq cartan iso}
Cartan connections $\varpi$ on $\calP$ are in one-to-one correspondence with isomorphisms of $C^\infty(\calM)$-modules $\varpiLie : \Gamma_H(\calP) \to \Gamma_H(\calP, \frakg)$ such that \eqref{eq-diagram-triangle-upleft} commutes. 
\end{theorem}

Recall that the curvature of a Cartan connection $\varpi$ is $\bOmega = \dd \varpi + \tfrac{1}{2} [\varpi, \varpi] \in \Omega^2(\calP) \otimes \frakg$. Notice that $\varpiLie \in \Omega^1_\lie(\calP, \frakg)$.

\begin{proposition}
$\bOmegaLie \defeq \hd \varpiLie - \tfrac{1}{2} [\varpiLie, \varpiLie] \in \Omega^2_\lie(\calP, \frakg)$ satisfies:
\begin{enumerate}[label=(\roman*)]
\item\label{item-bOmegalie-restriction} $\bOmegaLie$ is the restriction of $- \bOmega$ to $\Gamma_H(\calP) \subset \Gamma(T\calP)$.

\item\label{item-bOmega-iotaP} $\bOmegaLie$ vanishes on $\iota_\calP(\Gamma_H(\calP, \frakh))$. 
\end{enumerate}
\end{proposition}

\begin{proof}
\ref{item-bOmegalie-restriction}: For any $\frakX, \frakY \in \Gamma_H(\calP)$, one has
\begin{align*}
\bOmegaLie(\frakX, \frakY)
&= \frakX \cdot \varpiLie(\frakY) - \frakY \cdot \varpiLie(\frakX) - \varpiLie([\frakX, \frakY]) - [\varpiLie(\frakX), \varpiLie(\frakY)]
\\
&= - \frakX \cdot \varpi(\frakY) + \frakY \cdot \varpi(\frakX) + \varpi([\frakX, \frakY]) - [\varpi(\frakX), \varpi(\frakY)]
\\
&
= - \bOmega(\frakX, \frakY)
\end{align*}

\ref{item-bOmega-iotaP}: let $v \in \Gamma_H(\calP, \frakh)$ and $\frakX \in \Gamma_H(\calP)$, then, using $\varpiLie \circ \iota_\calP = i$ and Lemma~\ref{lem-actions-vector-fields},
\begin{align*}
\bOmegaLie(\frakX, \iota_\calP(v))
&= \frakX \cdot \varpiLie(\iota_\calP(v)) - \iota_\calP(v) \cdot \varpiLie(\frakX) - \varpiLie([\frakX, \iota_\calP(v)]) - [\varpiLie(\frakX), \varpiLie(\iota_\calP(v))]
\\
&= \frakX \cdot i(v) - [i(v), \varpiLie(\frakX)] - i(\frakX \cdot v) - [\varpiLie(\frakX), i(v)] 
= 0.
\end{align*}
\end{proof}

Using notations from Sect.~\ref{sec-local-trivializations}, let $\bA \defeq s^* \varpi$ be the local trivialization of the Cartan connection $1$-form. Then, using \eqref{eq-S-notnabla-Psi} and $p = s(x)h$,
\begin{align*}
\varpiLie \circ S(X \oplus \gamma)(p)
&= - \varpi_{| s(x)h} \left( T_{s(x)} R^\calP_h \circ T_x s (X_{| x}) - [\Ad_{h^{-1}} \circ \gamma(x)]^v_{| s(x) h} \right)
= - \Ad_{h^{-1}} \left( \varpi_{| s(x)}( T_x s (X_{| x}) ) - \gamma(x) \right)
\\
&= - \Ad_{h^{-1}} \left(  \bA(X)_{| x} - \gamma(x) \right)
= - \Psi_\frakg \left(  \bA(X) - \gamma \right)(s(x)h).
\end{align*} 
This gives
\begin{equation}
\label{eq-varpilie-loc}
\varpiLie \locto - \bA + i_\frakh \circ \theta_\frakh,
\end{equation}
which generalizes \eqref{eq-local-omegalie}, and the local trivializations of the isomorphisms $\varpiLie : \Gamma_H(\calP) \xrightarrow{\simeq} \Gamma_H(\calP, \frakg)$ and $\tvarpiLie : \Gamma(T\calM) \xrightarrow{\simeq} \Gamma_H(\calP, \frakg/\frakh)$ are
\begin{align*}
\Gamma(T\calU \oplus \calU \times \frakh) &\to \Gamma(\calU \times \frakg), & X \oplus \gamma &\mapsto -\bA(X) + i_\frakh \circ \gamma,
\\
\Gamma(T\calU) &\to \Gamma(\calU \times \frakg/\frakh), & X &\mapsto - [\bA(X)].
\end{align*}

Furthermore, on extending standard result in the literature~\cite[Appendix~A, §~3]{Shar97a} to Atiyah Lie algebroids, we show that specific Ehresmann connections on $\calQ$ provide Cartan connections on $\cal P$:
\begin{proposition}
\label{prop Ehresmann to Cartan Lie}
Let $\omega^\calQ \in \Omega^1(\calQ) \otimes \frakg$ be a connection $1$-form on $\calQ$, let $\nabla^\calQ : \Gamma(T\calM) \to \Gamma_G(\calQ)$ be the corresponding horizontal lift, and let $\omegaLie^\calQ \in \Omega^1_\lie(\calQ, \frakg)$ be its associated $1$-form  which all together define the short exact sequence
\begin{equation}
\label{eq-diagram-ses-omegaQLie}
\begin{tikzcd}[column sep=25pt, row sep=30pt]
\algzero \arrow[r] 
& \Gamma(T\calM) \arrow[r, "\nabla^\calQ"]
& \Gamma_G(\calQ) \arrow[r, "\omegaLie^\calQ"] 
& \Gamma_G(\calQ, \frakg) \arrow[r]
& \algzero\, .
\end{tikzcd}
\end{equation}
Then $\omega^\calQ$ defines a Cartan connection $\varpi$ on $\calP$ if and only if the map $R\circ\nabla^\calQ : \Gamma(T\calM) \to \Gamma_H(\calP, \frakg/\frakh)$ is an isomorphism of $C^\infty(\calM)$-modules. In the correspondence given by Theorem~\ref{thm eq cartan iso}, this Cartan connection $\varpi$ is related to the isomorphism 
\begin{align}
\label{eq varpiLie from omegaLieQ}
\varpiLie \defeq \hi^{-1} \circ \omegaLie^\calQ \circ J : \Gamma_H(\calP) \to \Gamma_H(\calP, \frakg),
\end{align}
for which the diagram \eqref{eq-diagram-triangle-upleft} commutes.

Moreover, the isomorphism $R\circ\nabla^\calQ$ implies 
\begin{equation}
\label{eq-ker omega cap im J = 0}
\ker \omegaLie^\calQ \cap J(\Gamma_H(\calP)) = \{ 0 \}.
\end{equation}
\end{proposition}

\begin{proof}
Let us first prove that if $R\circ\nabla^\calQ$ is an isomorphism then \eqref{eq-ker omega cap im J = 0} holds true. Since $\ker \omegaLie^\calQ = \im\nabla^\calQ$ and $\im J = \ker R$ by exactness of \eqref{eq-sec-J-R} and \eqref{eq-diagram-ses-omegaQLie},  \eqref{eq-ker omega cap im J = 0} is equivalent to $\im\nabla^\calQ \cap \ker R = \{ 0 \}$. Let $\hfX \in \im\nabla^\calQ \cap \ker R \subset \Gamma_G(\calQ)$: then, there is a $X\in \Gamma(T\calM)$ such that $\hfX = \nabla^\calQ_X$ and $R\circ\nabla^\calQ_X = 0$. Since $R\circ\nabla^\calQ$ is an isomorphism, this implies $X = 0$, and so $
\hfX = 0$, which proves that $\im\nabla^\calQ \cap \ker R = \{ 0 \}$.

From now on, we use Theorem~\ref{thm eq cartan iso} to work directly with $\varpiLie$ defined by \eqref{eq varpiLie from omegaLieQ}. The commutativity of \eqref{eq-diagram-triangle-upleft} is proved by a direct computation using the commutativity of \eqref{eq-diagram}: for any $v \in \Gamma_H(\calP, \frakh)$, one has
\begin{align}
\label{eq varpiLie iotaP v}
\varpiLie \circ \iota_\calP(v)
&= \hi^{-1} \circ \omegaLie^\calQ \circ J \circ \iota_\calP(v)
= \hi^{-1} \circ \omegaLie^\calQ \circ \iota_\calQ \circ j(v)
= \hi^{-1} \circ j(v)
= i(v).
\end{align}
Notice that the injectivity of $\varpiLie$ is a direct consequence of \eqref{eq-ker omega cap im J = 0}: let $\frakX \in \Gamma_H(\calP)$ be such that $\varpiLie(\frakX) = 0$, then $\omegaLie^\calQ \circ J(\frakX) = 0$, and so $J(\frakX) \in \ker \omegaLie^\calQ \cap J(\Gamma_H(\calP)) = \{ 0 \}$, so that $\frakX = 0$. On the contrary, the surjectivity of $\varpiLie$ is not so straightforward.\footnote{It seems to us that the condition \eqref{eq-ker omega cap im J = 0} alone, which ought to parallel the condition given in \cite[Appendix~A, Prop.~3.1]{Shar97a} on the Ehresmann connection on $\calQ$, is not “strong enough” to deal with the module aspect of our construction to show this surjectivity even if $\varpiLie$ is injective between two isomorphic modules.}

Let $\mromegaLie^\calP : \Gamma_H(\calP) \to \Gamma_H(\calP, \frakh)$ be a background (Ehresmann) connection on $\calP$, let $\mrnabla^\calP : \Gamma(T\calM) \to \Gamma_H(\calP)$ be its associated lift. From Prop.~\ref{prop-connextion-P-to-Q}, let $\mromegaLie^\calQ : \Gamma_G(\calQ) \to \Gamma_G(\calQ, \frakg)$ be the transported connection and $\mrnabla^\calQ : \Gamma(T\calM) \to \Gamma_G(\calQ)$ its associated lift. Then there is $\Delta^\calQ : \Gamma_G(\calQ) \to \Gamma_G(\calQ, \frakg)$ such that 
\begin{align}
\label{eq-compare-omegaLie}
\Delta^\calQ = \mromega_\lie^\calQ - \omegaLie^\calQ
\end{align}
with $\Delta^\calQ \circ \iota_\calQ = 0$ owing to the normalization of the connection 1-forms. Hence, for any $\hfX \in \Gamma_G(\calQ)$ with $X=\rho_\calQ(\hfX)$, one has the two decompositions
\begin{equation}
\label{eq hfX two decompositions}
\hfX = \mathring{\nabla}^\calQ_X + \iota_\calQ \circ \mromega_\lie^\calQ (\hfX) = \nabla^\calQ_X + \iota_\calQ \circ \omegaLie^\calQ (\hfX).
\end{equation}
Inserting these two decompositions into \eqref{eq-compare-omegaLie}, and using the normalization condition and the respective splittings, one gets
\begin{align}
\label{eq-Delta-Q}
\Delta^\calQ (\hfX) 
= \mromega_\lie^\calQ \circ \nabla^\calQ_X 
= - \omegaLie^\calQ \circ \mathring{\nabla}^\calQ_X 
\quad\text{\textsl{i.e.}}\quad 
\Delta^\calQ 
= \mromega_\lie^\calQ \circ \nabla^\calQ \circ \rho_\calQ 
= - \omegaLie^\calQ \circ \mathring{\nabla}^\calQ \circ \rho_\calQ .
\end{align}
Similarly to \eqref{eq-diag-five}, let us consider the diagram
\begin{equation} \label{eq-diag-five-bis}
\begin{tikzcd}[column sep=30pt, row sep=25pt]
0 \arrow[r] \arrow[d, equal] 
& \Gamma_H(\calP, \frakh) \arrow[r, "\iota_\calP"] \arrow[d, equal] 
& \Gamma_H(\calP) \arrow[r, "\rho_\calP"]
\arrow[d, "\varpiLie"] 
& \Gamma(T\calM) \arrow[r] \arrow[d, "-R\circ\nabla^\calQ"]
& 0 \arrow[d, equal] 
\\
 0 \arrow[r] 
 & \Gamma_H(\calP, \frakh) \arrow[r, "i"] 
 & \Gamma_H(\calP, \frakg) \arrow[r, "r"] 
 & \Gamma_H(\calP, \frakg/\frakh)\arrow[r] 
 & 0
\end{tikzcd}
 \end{equation}
The non trivial left square commutes by \eqref{eq varpiLie iotaP v}. Let us show that the non trivial right square commutes also. Any $\frakX\in \Gamma_H(\calP)$ can be decomposed as $\frakX = \mathring{\nabla}^\calP_X + \iota_\calP \circ \mromega_\lie^\calP (\frakX)$. Then, using \eqref{eq varpiLie iotaP v}, one has $\varpiLie(\frakX) = \varpiLie\circ \mathring{\nabla}^\calP_X + i\circ \mromega_\lie^\calP (\frakX)$ with
\begin{align*}
\varpiLie\circ \mathring{\nabla}^\calP_X 
& = \hi^{-1}\circ \omegaLie^\calQ \circ J \circ \mathring{\nabla}^\calP_X 
= \hi^{-1}\circ \omegaLie^\calQ \circ \mathring{\nabla}^\calQ_X 
= - \hi^{-1}\circ \mromega_\lie^\calQ \circ \nabla^\calQ_X
\end{align*}
where \eqref{eq varpiLie from omegaLieQ}, \eqref{eq-J-nabla-P-nabla-Q} (for the background connections), and \eqref{eq-Delta-Q} have been used successively. Then using \eqref{eq-sec-i-r} and the properties of \eqref{eq-diagram-ehresmann} (for the background connections)
\begin{equation*}
(r\circ\varpiLie)(\frakX) 
= r\circ\varpiLie\circ \mathring{\nabla}^\calP_X 
= - \hr\circ \mromega_\lie^\calQ \circ \nabla^\calQ_X 
= - R\circ \nabla^\calQ_X 
= - R\circ \nabla^\calQ \circ \rho_\calQ(\frakX) .
\end{equation*}
Hence, the diagram \eqref{eq-diag-five-bis} is commutative and exact. Using the five lemma in this diagram, there is a equivalence between $\varpiLie$ and $-R\circ \nabla^\calQ$ being isomorphisms. In particular, when $\varpiLie$ is a Cartan connection, then $\tvarpiLie = - R\circ \nabla^\calQ$.
\end{proof}

\subsection{Cartan connections: the reductive case}
\label{sec Cartan reductive}

Following \cite{Shar97a}, a Cartan geometry is said to be \emph{reductive} if the Lie algebra $\frakg$ admits a decomposition $\frakg=\frakh \oplus \frakp$ where $\frakp$ is a $\frakh$-module. In that case, denote by $\pi_\frakh : \frakg \to \frakh$ and $\pi_\frakp : \frakg \to \frakp$ the corresponding projections. We identify the $\frakh$-modules $\frakg/\frakh$ and $\frakp$ in a natural way, and the projections induce projections $\pi_\frakh : \Gamma_H(\calP, \frakg) \to \Gamma_H(\calP, \frakh)$ and $\pi_\frakp : \Gamma_H(\calP, \frakg) \to \Gamma_H(\calP, \frakp) \simeq \Gamma_H(\calP, \frakg/\frakh)$. The inclusion $i_\frakp : \frakp \hookrightarrow \frakg$ induces an inclusion $\bi : \Gamma_H(\calP, \frakp) \hookrightarrow \Gamma_H(\calP, \frakg)$. Notice that
\begin{align}
\label{eq-pifrakh-i-pifrakp-prifrakp-bi}
\pi_\frakh \circ i &= \Id,
&
\pi_\frakp &= r,
&
\pi_\frakp \circ \bi &= \Id.
\end{align}

A Cartan connection $\varpi$ splits as $\varpi = i_\frakh \circ \omega + i_\frakp \circ \beta$ where $\omega \defeq \pi_\frakh \circ \varpi$ and $\beta \defeq \pi_\frakp \circ \varpi$. Its associated $1$-form $\varpiLie$ splits also as $\varpiLie = i \circ \omegaLie \oplus \bi \circ \betaLie$ where
\begin{align*}
\omegaLie &\defeq \pi_\frakh \circ \varpiLie : \Gamma_H(\calP) \to \Gamma_H(\calP, \frakh), 
&
\betaLie &\defeq \pi_\frakp \circ \varpiLie : \Gamma_H(\calP) \to  \Gamma_H(\calP, \frakp).
\end{align*}
Then $\omega$ is an Ehresmann connection $1$-form on $\calP$ and $\omegaLie$ is its associated $1$-form: we denote by $\nabla^\omega : \Gamma(T\calM) \to \Gamma_H(\calP)$ its horizontal lift.

\begin{proposition}
One has $\omegaLie \circ \iota_\calP = \Id$, $\betaLie \circ \iota_\calP = 0$ and, for any $\frakX \in \Gamma_H(\calP)$ with $X = \rho_\calP(\frakX)$, $\frakX = \nabla^\omega_X + \iota_\calP \circ \omegaLie(\frakX)$ and $\tvarpiLie(X) = \betaLie(\frakX) = \betaLie \circ \nabla^\omega_X$.
\end{proposition}

\begin{proof}
One has $\omegaLie \circ \iota_\calP = \pi_\frakh \circ \varpiLie \circ \iota_\calP = \pi_\frakh \circ i = \Id$ and $\betaLie \circ \iota_\calP = \pi_\frakp \circ \varpiLie \circ \iota_\calP = \pi_\frakp \circ i = 0$.

Let $\alpha^\omega \in \Omega^1_\lie(\calP, \frakh)$ such that $\frakX = \nabla^\omega_X - \iota_\calP \circ \alpha^\omega(\frakX)$ \cite[Prop.~3.9]{LazzMass12a}. Then, using $\omega \circ \nabla^\omega = 0$, we have $\omegaLie(\frakX) = - \pi_\frakh \circ \varpi(\nabla^\omega_X) + \pi_\frakh \circ \varpi \circ \iota_\calP \circ \alpha^\omega(\frakX) = 0 - \omegaLie \circ \iota_\calP \circ \alpha^\omega(\frakX) = - \alpha^\omega(\frakX)$.

One has $\tvarpiLie(X) = r \circ \varpiLie(\frakX) = r \circ i \circ \omegaLie(\frakX) + r \circ \bi \circ \betaLie(\frakX) = \betaLie(\frakX)$, and since $\betaLie \circ \iota_\calP = 0$ one gets $\betaLie(\frakX) = \betaLie \circ \nabla^\omega_X$.
\end{proof}

So, we get the following maps of $C^\infty(\calM)$-modules:
\begin{equation}
\label{eq-diagram-cartan-reductive}
\begin{tikzcd}[column sep=38pt, row sep=30pt]
\Gamma_H(\calP, \frakh) 
	\arrow[r, "\iota_\calP"] 
	\arrow[d, "i", xshift=4pt] 
& |[alias=GHP]| \Gamma_H(\calP) 
	\arrow[r, "\rho_\calP"] 
	\arrow[d, "J" ] 
	\arrow[dl, "\varpiLie",  "\simeq"']
	\arrow[l, bend right=25, "\omegaLie"']
& \Gamma(T\calM) 
	\arrow[d, equal] 
	\arrow[l, bend right=25, "\nabla^\omega"']
\\
\Gamma_H(\calP, \frakg) 
	\arrow[r, "\iota_\calQ"] 
	\arrow[d,  "r = \pi_\frakp", xshift=4pt]
	\arrow[u, "\pi_\frakh",  xshift=-4pt]
& \Gamma_G(\calQ) 
	\arrow[r, "\rho_\calQ"] 
	\arrow[d,  "R"]
& \Gamma(T\calM) 
	\arrow[dl, "\tvarpiLie",  "\simeq"'] 
\\
 |[alias=GHPp]| \Gamma_H(\calP, \frakp) 
	\arrow[u, "\bi",  xshift=-4pt]
 	\arrow[r, equal]
& \Gamma_H(\calP, \frakp)
& 
\arrow[from=GHP, to=GHPp, bend left=13, crossing over, crossing over clearance=0.7ex, pos=0.7, "\betaLie"]
\end{tikzcd}  
\end{equation}
Notice that this diagram is \emph{not} commutative everywhere.

\smallskip
Denote by $\bA = A \oplus B = s^* \varpi$ the local trivialization of $\varpi$, where $A$ takes its values in $\frakh$ and $B$ in $\frakp$. Then $A = s^* \omega$ and $B = s^* \beta$ are the local trivializations of the components of $\varpi$  (strictly speaking, $s^* \omega = \pi_\frakh \circ A$ and $s^* \beta = \pi_\frakp \circ B$ but we omit the projections). Then, using \eqref{eq-varpilie-loc}, one gets
\begin{align}
\label{eq-omegaLie-betaLie-varpiLie-loc}
\varpiLie &\locto -(A + B) + i_\frakh \circ \theta_\frakh.
&
\omegaLie &\locto - A + \theta_\frakh,
&
\betaLie &\locto -B,
\end{align}
and the trivialization of the map $\varpiLie$ is $X \oplus \gamma \mapsto \gamma - A(X) - B(X)$, of the map $\omegaLie$ is $X \oplus \gamma \mapsto \gamma - A(X)$, and of the map $\betaLie$ is $X \oplus \gamma \mapsto -B(X)$. In all these expressions, the inclusions of $\frakh$ and $\frakp$ in $\frakg$ are omitted.

\subsection{Cartan connections and metrics}

Let us recall the definition and the main results concerning the notion of metric on $\Gamma_H(\calP)$ as given in \cite[Sect.~2]{FourLazzMass13a}. A metric on $\Gamma_H(\calP)$ is a $C^\infty(\calM)$-linear symmetric map 
\begin{equation*}
\hg : \Gamma_H(\calP) \otimes_{C^\infty(\calM)} \Gamma_H(\calP) \to C^\infty(\calM).
\end{equation*}
Such a metric defines a inner metric $h \defeq \iota_\calP^* \hg : \Gamma_H(\calP, \frakh) \otimes_{C^\infty(\calM)} \Gamma_H(\calP, \frakh) \to C^\infty(\calM)$, with $h(v,w) = \hg(\iota_\calP(v), \iota_\calP(w))$. The metric is said to be inner non degenerate if $h$ is non degenerated as a metric on the vector bundle $\calL^\calP$. Given a metric $g$ on $\calM$, $\rho_\calP^* g$ defines a metric on $\Gamma_H(\calP)$: $(\rho_\calP^* g)(\frakX, \frakY) = g(\rho_\calP(\frakX), \rho_\calP(\frakY))$.

We will use the following characterization of inner non degenerate metric on $\Gamma_H(\calP)$:
\begin{proposition}[{\cite[Prop.~2.7]{FourLazzMass13a}}]
An inner non degenerate metric $\hg$ on $\Gamma_H(\calP)$ is equivalent to a triple $(g, h, \nabla^\calP)$ where $g$ is a metric on $\calM$, $h$ is an inner non degenerate metric on $\Gamma_H(\calP, \frakh)$ and $\nabla^\calP : \Gamma(T\calM) \to \Gamma_H(\calP)$ is a connection, with
\begin{equation*}
\hg(\frakX, \frakY) = g(X, Y) + h(v, w) 
\quad \text{for} \quad 
\frakX = \nabla^\calP_X + \iota_\calP(v), \frakY = \nabla^\calP_Y + \iota_\calP(w).
\end{equation*}
The connection $\nabla^\calP$ is uniquely determined by the condition
\begin{equation}
\label{eq-def-nabla-metric}
\hg(\nabla^\calP_X, \iota_\calP(v)) = 0 \text{ for any $X \in \Gamma(T\calM)$ and $v \in \Gamma_H(\calP, \frakh)$}.
\end{equation}
\end{proposition}

Let $\hh : \Gamma_H(\calP, \frakg) \otimes_{C^\infty(\calM)} \Gamma_H(\calP, \frakg) \to C^\infty(\calM)$ be a fixed non degenerate metric. Then one has the orthogonal decomposition $\Gamma_H(\calP, \frakg) = i(\Gamma_H(\calP, \frakh)) \oplus [i(\Gamma_H(\calP, \frakh))]^\perp$.

\begin{proposition}
\label{prop cartan connection metric projections}
A Cartan connection $\varpi$ defines a unique inner non degenerate metric $\hg \defeq \varpiLie^* \hh$ on $\Gamma_H(\calP)$. Let  $(g, h, \nabla^\calP)$ be its associated triple with $\omegaLie^\calP \in \Omega^1_\lie(\calP, \frakh)$ the $1$-form associated to $\nabla^\calP$ as in \eqref{eq-diagram-ehresmann}. Then $h$ is the restriction of $\hh$ to $\Gamma_H(\calP, \frakh) \xhookrightarrow{i} \Gamma_H(\calP, \frakg)$, so that $h = i^* \hh$, and the $1$-form $\omegaLie^\calP$ is uniquely defined by
\begin{equation*}
h( \omegaLie^\calP(\frakX), v) = \hh( \varpiLie(\frakX), i(v)) 
\quad \text{for any} \quad
\frakX \in \Gamma_H(\calP), v \in \Gamma_H(\calP, \frakh).
\end{equation*}

The map $p_\frakh \defeq i \circ \omegaLie^\calP \circ \varpiLie^{-1} : \Gamma_H(\calP, \frakg) \to \Gamma_H(\calP, \frakg)$ is independent of the choice of the Cartan connection $\varpi$ and is the orthogonal projection onto $i(\Gamma_H(\calP, \frakh))$. So the map $1 - p_\frakh$ is the orthogonal projection onto $[i(\Gamma_H(\calP, \frakh))]^\perp$ and
\begin{equation}
\label{eq-im-varpiLie-i-omegaLie-perp-i}
\im ( \varpiLie - i \circ \omegaLie^\calP ) = [i(\Gamma_H(\calP, \frakh))]^\perp.
\end{equation}

In the orthogonal decomposition $\varpiLie \rdefeq \varpiLie^\frakh \oplus \varpiLie^\perp$ along $i(\Gamma_H(\calP, \frakh)) \oplus [i(\Gamma_H(\calP, \frakh))]^\perp$, one has $\varpiLie^\frakh = i \circ \omegaLie^\calP$.
\end{proposition}

Notice that, since $\im (1 - p_\frakh) =  [i(\Gamma_H(\calP, \frakh))]^\perp$, one has
\begin{equation*}
\hh(i(v), p_\frakh(\hw)) = \hh(i(v), \hw) \quad \text{for any } v \in \Gamma_H(\calP, \frakh), \hw \in \Gamma_H(\calP, \frakg).
\end{equation*}

\begin{proof}
By definition,  for any $\frakX, \frakY \in \Gamma_H(\calP)$, it is straightforward to verify that $\hg(\frakX, \frakY) \defeq \hh(\varpiLie(\frakX), \varpiLie(\frakY))$ defines a metric on $\Gamma_H(\calP)$. Then, for any $v, w \in \Gamma_H(\calP, \frakh)$, one has $\hg(\iota_\calP(v), \iota_\calP(w)) = \hh(\varpiLie \circ \iota_\calP(v), \varpiLie \circ \iota_\calP(w)) = \hh( i(v), i(w))$, which implies $h = \iota_\calP^* \hg = i^* \hh$, so that $h$ is the restriction of $\hh$ to $\Gamma_H(\calP, \frakh)$. This also implies that $\hg$ is inner non degenerate and so defines a triple $(g, h, \nabla^\calP)$. 

Following the proof of \cite[Prop.~2.6]{FourLazzMass13a} with the new convention (sign of $\omegaLie^\calP$), $\omegaLie^\calP$ is uniquely defined by the relation $h( \omegaLie^\calP(\frakX), v) = \hg(\frakX, \iota_\calP(v)) = \hh( \varpiLie(\frakX), i(v))$ for any $\frakX \in \Gamma_H(\calP)$ and $v \in \Gamma_H(\calP, \frakh)$, which can be written as $\hh( \varpiLie(\frakX) - i \circ \omegaLie^\calP(\frakX), i(v)) = 0$, so that $\im ( \varpiLie - i \circ \omegaLie^\calP ) \subset [i(\Gamma_H(\calP, \frakh))]^\perp$.

Let $\varpi'$ be a second Cartan connection which defines the maps $\varpiLie'$ and $\omegaLie'^\calP$. Then there exit two maps $\alpha : \Gamma(T\calM) \to \Gamma_H(\calP, \frakg)$ and $\beta : \Gamma_H(\calP, \frakg/\frakh) \to \Gamma_H(\calP)$ such that
\begin{align*}
\varpiLie' &= \varpiLie + \alpha \circ \rho_\calP,
&
\varpi_\lie'^{-1} &= \varpiLie^{-1} + \beta \circ r.
\end{align*}
Defining $\alpha^\frakh$ by the  orthogonal decomposition $\alpha = i \circ \alpha^\frakh \oplus \alpha^\perp$, one has
\begin{align*}
\omegaLie'^\calP &= \omegaLie^\calP + \alpha^\frakh \circ \rho_\calP,
&
\nabla'^\calP &= \nabla^\calP - \iota_\calP \circ \alpha^\frakh,
\end{align*}
from which we deduce that $\varpiLie' \circ \nabla'^\calP = \varpiLie \circ \nabla^\calP + \alpha^\perp$. Since $\varpiLie' \circ \varpi_\lie'^{-1} = \Id$, one gets the constrain
\begin{equation*}
\varpiLie \circ \beta \circ r  + \alpha \circ \rho_\calP \circ \varpiLie^{-1} + \alpha \circ \rho_\calP \circ \beta \circ r =0.
\end{equation*}
Once projected on $\frakh$, this constrain reduces to
\begin{equation*}
\omegaLie^\calP \circ \beta \circ r + \alpha^\frakh \circ \rho_\calP \circ \varpiLie^{-1} + \alpha^\frakh \circ \rho_\calP \circ \beta \circ r =0,
\end{equation*}
and implies that
\begin{align*}
\omegaLie'^\calP \circ \varpi_\lie'^{-1} 
&= ( \omegaLie^\calP + \alpha^\frakh \circ \rho_\calP) \circ (\varpi_\lie^{-1} + \beta \circ r)
= \omegaLie^\calP \circ \varpiLie^{-1}
\end{align*}
So the map  $p_\frakh = i \circ \omegaLie^\calP \circ \varpiLie^{-1}$ is independent of the choice of $\varpi$. Let us show that it is a projection:
\begin{align*}
p_\frakh \circ p_\frakh
&= i \circ \omegaLie^\calP \circ \varpiLie^{-1} \circ i \circ \omegaLie^\calP \circ \varpiLie^{-1}
= i \circ \omegaLie^\calP \circ \iota_\calP \circ \omegaLie^\calP \circ \varpiLie^{-1}
= i \circ  \omegaLie^\calP \circ \varpiLie^{-1}
= p_\frakh.
\end{align*}
By construction, $\im p_\frakh \subset i(\Gamma_H(\calP, \frakh))$ and one has $p_\frakh \circ i = i \circ \omegaLie^\calP \circ \varpiLie^{-1} \circ i = i \circ \omegaLie^\calP \circ \iota_\calP = i$, so that, for any $v \in \Gamma_H(\calP, \frakh)$, $i(v) = p_\frakh \circ i(v) \in \im p_\frakh$, from which we conclude that $\im p_\frakh = i(\Gamma_H(\calP, \frakh))$. From the previous result we know that $\im (1-p_\frakh) = \im ( \varpiLie - i \circ \omegaLie^\calP ) \circ \varpiLie^{-1} \subset [i(\Gamma_H(\calP, \frakh))]^\perp$. From the decompositions $\Gamma_H(\calP, \frakg) = i(\Gamma_H(\calP, \frakh)) \oplus [i(\Gamma_H(\calP, \frakh))]^\perp$ and $\Gamma_H(\calP, \frakg) = \im p_\frakh \oplus \im (1-p_\frakh)$, and since $\im p_\frakh = i(\Gamma_H(\calP, \frakh))$, we get the equality $\im (1-p_\frakh) = [i(\Gamma_H(\calP, \frakh))]^\perp$, and so $\im ( \varpiLie - i \circ \omegaLie^\calP ) = [i(\Gamma_H(\calP, \frakh))]^\perp$.

One has $p_\frakh \circ \varpiLie = i \circ \omegaLie^\calP \circ \varpiLie^{-1} \circ \varpiLie = i \circ \omegaLie^\calP$, which proves that $i \circ \omegaLie^\calP$ is the orthogonal projection of $\varpiLie$ on $i(\Gamma_H(\calP, \frakh))$.
\end{proof}

\begin{proposition} \label{prop-indep}
The map
\begin{equation*}
\Theta^{\hh} \defeq \omegaLie^\calP \circ \varpiLie^{-1} : \Gamma_H(\calP, \frakg) \to \Gamma_H(\calP, \frakh)
\end{equation*}
is independent of the choice of the Cartan connection $\varpi$ and defines an orthogonal splitting of the s.e.s. \eqref{eq-sec-i-r}, with $\Theta^{\hh} \circ i = \Id$ and $\Theta^{\hh} \circ \varpiLie^\perp = 0$.

To $\Theta^{\hh}$ we associate the map $\bnabla^{\hh} : \Gamma_H(\calP, \frakg/\frakh) \to \Gamma_H(\calP, \frakg)$ defined, for any $\bv \in \Gamma_H(\calP, \frakg/\frakh)$, by 
\begin{equation}
\label{eq-def-bnabla}
\bnabla^{\hh} (\bv) \defeq \hv - i \circ \Theta^{\hh}(\hv)
\quad\text{for any $\hv \in \Gamma_H(\calP, \frakg)$ such that $r(\hv) = \bv$}.
\end{equation}
This map defines a metric $\bh : \Gamma_H(\calP, \frakg/\frakh) \otimes_{C^\infty(\calM)} \Gamma_H(\calP, \frakg/\frakh) \to C^\infty(\calM)$ by
\begin{equation*}
\bh(\bv, \bw) \defeq \hh(\bnabla^{\hh}(\bv), \bnabla^{\hh}(\bw)) 
\quad \text{for any} \quad
\bv, \bw \in \Gamma_H(\calP, \frakg/\frakh).
\end{equation*}
Then one has $\bnabla^{\hh} \circ \tvarpiLie = \varpiLie \circ \nabla^\calP$, $g = \tvarpiLie^* \bh$, and $\im \bnabla^{\hh} = [i(\Gamma_H(\calP, \frakh))]^\perp$.
\end{proposition}

\begin{proof}
We have shown in the proof of Prop.~\ref{prop cartan connection metric projections} that $\Theta^{\hh}$, and then $\bnabla^{\hh}$ and $\bh$, are independent of $\varpi$.

One has $\Theta^{\hh} \circ i = \omegaLie^\calP \circ \varpiLie^{-1} \circ i = \omegaLie^\calP \circ \iota_\calP = \Id$ and, from $\varpiLie^\perp = \varpiLie - i \circ \omegaLie^\calP$, $\Theta^{\hh} \circ \varpiLie^\perp = \omegaLie^\calP \circ \varpiLie^{-1} \circ (\varpiLie - i \circ \omegaLie^\calP) = \omegaLie^\calP - \omegaLie^\calP \circ \iota_\calP \circ \omegaLie^\calP = 0$.

Let us show that the map $\bnabla^{\hh}$ is well defined. If $\hv, \hv' \in \Gamma_H(\calP, \frakg)$ are such that $r(\hv) = \bv$ and $r(\hv') = \bv$, then there exists $v \in \Gamma_H(\calP, \frakh)$ such that $\hv' = \hv + i(v)$. Then $\hv' - i \circ \Theta^{\hh}(\hv') = \hv + i(v) - i \circ \Theta^{\hh}(\hv) - i \circ \omegaLie^\calP \circ \varpiLie^{-1} \circ i(v) = \hv - i \circ \Theta^{\hh}(\hv) + i(v) - i \circ \omegaLie^\calP \circ \iota_\calP(v) = \hv - i \circ \Theta^{\hh}(\hv)$.

Let $X \in \Gamma(T\calM)$ and $\frakX  \in \Gamma_H(\calP)$ such that $\rho_\calP(\frakX) = X$. Using $\frakX = \nabla^\calP_X + \iota_\calP \circ \omegaLie^\calP(\frakX)$ and \eqref{eq-def-tvarpiLie}, one has
\begin{align*}
\bnabla^{\hh} \circ \tvarpiLie(X)
&= \bnabla^{\hh} \circ r \circ \varpiLie(\frakX)
= \varpiLie(\frakX) - i \circ \Theta^{\hh} \circ \varpiLie(\frakX)
= \varpiLie(\frakX) - i \circ \omegaLie^\calP(\frakX)
\\
&
= \varpiLie(\nabla^\calP_X) + \varpiLie \circ \iota_\calP \circ \omegaLie^\calP(\frakX) - i \circ \omegaLie^\calP(\frakX)
= \varpiLie(\nabla^\calP_X).
\end{align*} 
Now, from the various definitions, for any $X, Y \in \Gamma(T\calM)$, one has
\begin{align*}
g(X,Y)
&= \hg( \nabla^\calP_X, \nabla^\calP_Y )
= \hh( \varpiLie(\nabla^\calP_X), \varpiLie(\nabla^\calP_Y))
= \hh( \bnabla^{\hh} \circ \tvarpiLie(X), \bnabla^{\hh} \circ \tvarpiLie(Y))
= \bh( \tvarpiLie(X), \tvarpiLie(Y)).
\end{align*}
Since $\bnabla^{\hh} \circ r = 1 - p_\frakh$, $r$ is surjective, and $\im (1 - p_\frakh) =  [i(\Gamma_H(\calP, \frakh))]^\perp$, we get $\im \bnabla^{\hh} = [i(\Gamma_H(\calP, \frakh))]^\perp$.
\end{proof}

Hence, from the (non degenerate) metric $\hh$ and the Cartan connection $\varpi$ we have defined (transported) a inner non degenerate metric $\hg$ on $\Gamma_H(\calP)$ with associated triple $(g, h, \nabla^\calP)$, and maps $\omegaLie^\calP$, $\Theta^{\hh}$, and $\bnabla^{\hh}$. All these constructions imply the commutativity of the maps with solid lines in the following diagram and the splitting of its lower row is orthogonal for $\hh$.
\begin{equation}
\label{eq-diagram-splitting-metric}
\begin{tikzcd}[column sep=30pt, row sep=25pt]
\Gamma_H(\calP, \frakh) 
	\arrow[r, "\iota_\calP", dashed] 
	\arrow[d, equal] 
& \Gamma_H(\calP) 
	\arrow[l, bend right=25, "\omegaLie^\calP"']
	\arrow[r, "\rho_\calP", dashed] 
	\arrow[d,  "\varpiLie",  "\simeq"'] 
& \Gamma(T\calM) 
	\arrow[l, bend right=25, "\nabla^\calP"']
	\arrow[d, "\tvarpiLie",  "\simeq"'] 
\\
\Gamma_H(\calP, \frakh) 
	\arrow[r, "i", dashed] 
& \Gamma_H(\calP, \frakg) 
	\arrow[l, bend right=25, "\Theta^{\hh}"']
 	\arrow[r, "r", dashed] 
& \Gamma_H(\calP, \frakg/\frakh) 
	\arrow[l, bend right=25, "\bnabla^{\hh}"']
\end{tikzcd}
\end{equation}
Notice that the splitting of the lower row in \eqref{eq-diagram-splitting-metric} is given by canonical maps which do not depend on $\varpi$.

\medskip
When the Cartan geometry is \emph{reductive}, it is natural to ask for a compatibility condition between the metric $\hh$ and the decomposition $\frakg = \frakh \oplus \frakp$: we require that $i \circ \pi_\frakh$ and $\bi \circ \pi_\frakp$ are orthogonal projections in $\Gamma_H(\calP, \frakg)$, with $i \circ \pi_\frakh + \bi \circ \pi_\frakp = \Id$. This is equivalent to say that $\hh$ is defined with a metric on $\Gamma_H(\calP, \frakh)$ and a metric on $\Gamma_H(\calP, \frakp)$. So, by hypothesis $[i(\Gamma_H(\calP, \frakh))]^\perp = \Gamma_H(\calP, \frakp)$, and, from \eqref{eq-pifrakh-i-pifrakp-prifrakp-bi} and \eqref{eq-im-varpiLie-i-omegaLie-perp-i}, one gets $0 = \pi_\frakh \circ (\varpiLie - i \circ \omegaLie^\calP) = \omegaLie - \omegaLie^\calP$, to which we conclude that the associated (Ehresmann) connection $\nabla^\calP$ in the triple $(g, h, \nabla^\calP)$ is the (Ehresmann) connection $\nabla^\omega$ defined by the reductive decomposition of $\varpi$. This relates the metrics $g$ and $\bh$ by $g = (\betaLie \circ \nabla^\omega)^* \bh$ since $\tvarpiLie = \betaLie \circ \nabla^\omega$.

This implies also that $\Theta^{\hh} = \pi_\frakh$, since $\omegaLie^\calP \circ \varpiLie^{-1} = \omegaLie \circ \varpiLie^{-1} = \pi_\frakh \circ \varpiLie \circ \varpiLie^{-1} = \pi_\frakh$. Then $p_\frakh = i \circ \pi_\frakh$, so that $1 - p_\frakh = 1 - i \circ \pi_\frakh = \bi \circ \pi_\frakp$, from which we get $\bnabla^{\hh} = \bi$ since $r = \pi_\frakp$.

\subsection{Gauge transformations and diffeomorphisms}
\label{subsec-gaugeXdiff}

In \cite[Sect.~3]{LazzMass12a}, infinitesimal (inner) gauge transformations of connections defined on transitive Lie algebroids where identified as Lie derivatives along vector fields $\iota_\calP(v)$ for $v \in \Gamma_H(\calP, \frakh)$ (see \eqref{eq-def-lie-derivative}). For Atiyah Lie algebroids (the present situation), “finite” gauge transformations where also considered, since then the gauge group of the principal bundle acts on connections, and its Lie algebra identifies with $\Gamma_H(\calP, \frakh)$. Let us recall the result in \cite[Prop.~3.13]{LazzMass12a}: for any connection $1$-form $\omega$ on $\calP$ with associated $1$-form $\alpha \in \Omega^1_\lie(\calP, \frakh)$ such that $\frakX = \nabla^\omega - \iota_\calP \circ \alpha(\frakX)$ for any $\frakX \in \Gamma_H(\calP)$, and for any $v \in \Gamma_H(\calP, \frakh)$, one has
\begin{equation}
\label{eq-Lv-alpha}
L_v \alpha \defeq L_{\iota_\calP(v)} \alpha
= - \left( \hd_\frakh v + [\alpha, v] \right)
\end{equation}

We are going to show that the Lie derivative along any element $\frakX \in \Gamma_H(\calP)$ combines an inner gauge transformation and an infinitesimal diffeomorphism. Indeed, $\frakX \locto X \oplus \gamma$, where $\gamma : \calU \to \frakh$ is the generator of an infinitesimal gauge transformation, while $X$, as a vector field, is the generator of an infinitesimal diffeomorphism. 

Notice that applying the Lie derivative along $X \in \Gamma(T\calM)$ on elements in $\Omega^\grast_\lie(\calP, \frakh)$ or $\Omega^\grast_\lie(\calP, \frakg)$ does not make sense. We have to choose a lift $\frakX$ of $X$ in $\Gamma_H(\calP)$ for the Lie derivative to make sense. This necessarily  “add” an inner part in the game, \textit{i.e.} an inner gauge transformation.

Let $\omega$ be a connection $1$-form on $\calP$ and $\omegaLie \in \Omega^1_\lie(\calP, \frakh)$ its associated $1$-form as in \eqref{eq-diagram-ehresmann} (see also \eqref{eq-diagram-cartan-reductive}), then \eqref{eq-Lv-alpha} becomes
\begin{equation*}
L_v \omegaLie 
= \hd_\frakh v + [v, \omegaLie].
\end{equation*}

\begin{proposition}
\label{prop-lie-derivative-A-theta}
For any $\frakX \in \Gamma_H(\calP)$, the $1$-form $L_\frakX \omegaLie$ vanishes on $\iota_\calP(\Gamma_H(\calP, \frakh))$.

With $\omegaLie \locto - A + \theta_\frakh$ and $\frakX \locto X \oplus \gamma$, one has
\begin{equation}
\label{eq Lie derivative local omegaLie}
L_\frakX \omegaLie \locto - L_X A + \underbrace{\dd \gamma + [A, \gamma]}_{\rdefeq D_A \gamma = \delta_\gamma A},
\end{equation}
where the first term represents the action of the infinitesimal  diffeomorphism $X$ on $-A$ and the last two terms the action of the infinitesimal gauge transformation $\gamma$, identified to the covariant derivative of $\gamma$ along the gauge field $A$.
\end{proposition}

Here $D_A \gamma = \delta_\gamma A \defeq \dd \gamma + [A, \gamma]$ is the usual infinitesimal (gauge) action of $\gamma$ on the local connection $1$-form $A$.

\begin{proof}
For any $\frakY \in \Gamma_H(\calP)$, one has $(L_\frakX \omegaLie)(\frakY) = \frakX \cdot \omegaLie(\frakY) - \omegaLie([\frakX, \frakY])$. So, for any $v \in \Gamma_H(\calP, \frakh)$, one has $(L_\frakX \omegaLie)(\iota_\calP(v)) = \frakX \cdot \omegaLie \circ \iota_\calP(v) - \omegaLie([\frakX, \iota_\calP(v)]) = 0$ thanks to $\omegaLie \circ \iota_\calP(v) = v$ and $[\frakX, \iota_\calP(v)] = \iota_\calP(\frakX \cdot v)$.

With $\frakY \locto X \oplus \eta$ and $\omegaLie(\frakY) \locto -A(Y) + \eta$, one has 
\begin{align*}
\frakX \cdot \omegaLie(\frakY) & \locto
X \cdot (-A(Y) + \eta) + [\gamma, -A(Y) + \eta]
= -X \cdot A(Y) + X\cdot \eta - [\gamma, A(Y)] + [\gamma, \eta],
\\
- \omegaLie([\frakX, \frakY]) & \locto
A([X,Y]) - ( X\cdot \eta - Y \cdot \gamma + [\gamma, \eta])
\end{align*}
so that 
\begin{align*}
(L_\frakX \omegaLie)(\frakY) &\locto
-( X \cdot A(Y) - A([X,Y]) ) + Y \cdot \gamma + [A(Y), \gamma] 
= - (L_X A)(Y) + (\dd \gamma + [A, \gamma])(Y),
\end{align*}
where all contributions with $\eta$ cancel as expected.
\end{proof}

Now, let $\varpi$ be a Cartan $1$-form on $\calP$ and $\varpiLie \in \Omega^1_\lie(\calP, \frakg)$ its associated $1$-form as in \eqref{eq-def-varpiLie}.

\begin{proposition}
\label{prop-lie-derivative-varpiLie-loc}
For any $\frakX \in \Gamma_H(\calP)$, the $1$-form $L_\frakX \varpiLie$ vanishes on $\iota_\calP(\Gamma_H(\calP, \frakh))$.

With $\varpiLie \locto - \bA + i_\frakh \circ \theta_\frakh$ and $\frakX \locto X \oplus \gamma$, one has
\begin{equation*}
L_\frakX \varpiLie \locto - L_X \bA + \underbrace{\dd i_\frakh(\gamma) + [\bA, i_\frakh(\gamma)]}_{\rdefeq \delta_\gamma \bA},
\end{equation*}
where $\delta_\gamma \bA$ describes the infinitesimal gauge transformation of $\gamma$ on $\bA$ and $L_X \bA$ the infinitesimal diffeomorphism $X$ acting on $\bA$.
\end{proposition}

\begin{proof}
The proof is similar to the one for Prop.~\ref{prop-lie-derivative-A-theta}, using $\varpiLie \circ \iota_\calP = i$ and $\frakX \cdot i(v) = i(\frakX \cdot v)$, see \eqref{eq-actions-vector-fields}.
\end{proof}

\begin{corollary}
\label{cor-lie-derivative-varpiLie-reductive-loc}
If the Cartan geometry is reductive, then, with $\varpiLie \locto -(A + B) + i_\frakh \circ \theta_\frakh$ as in \eqref{eq-omegaLie-betaLie-varpiLie-loc}, one has
\begin{equation*}
L_\frakX \varpiLie \locto - L_X A - L_X B + \underbrace{\dd i_\frakh(\gamma) + [A, i_\frakh(\gamma)]}_{\rdefeq \delta_\gamma A} - \underbrace{[i_\frakh(\gamma), B]}_{\rdefeq \delta_\gamma B}.
\end{equation*}
\end{corollary}

Here $\delta_\gamma B \defeq [i_\frakh(\gamma), B]$ is the infinitesimal (gauge) action of $\gamma$ on the local $1$-form $B$ which transforms homogeneously under (finite) gauge transformations.

\medskip
Let $\mromega$ be a fixed (background) connection $1$-form on $\calP$, $\mrnabla$ its horizontal lift and let $\mromega \locto \mrA$. 
\begin{proposition}
\label{prop infinitezimal action}
For any $\frakX = \mrnabla_X + \iota_\calP(v) \in \Gamma_H(\calP)$, any $X \in \Gamma(T\calM)$, and any $v \in \Gamma_H(\calP, \frakh)$ such that $v \locto \gamma$, one has
\begin{equation*}
L_{\frakX} \varpiLie \locto - L_X \bA + \dd i_\frakh(\mrA(X)) + [\bA, i_\frakh(\mrA(X))] + \delta_\gamma \bA,
\end{equation*}
so that, in the \emph{reductive} case, one has
\begin{equation}
\label{eq-lie-derivative-varpiLie-reductive-frakX-loc}
L_{\frakX} \varpiLie \locto - L_X A + D_A i_\frakh(\mrA(X)) - L_X B - [i_\frakh(\mrA(X)), B]  + \delta_\gamma A - \delta_\gamma B.
\end{equation}
\end{proposition}

\begin{proof}
One has $\mrnabla_X + \iota_\calP(v) \locto X \oplus (\mrA(X) + \gamma)$ so that, using Prop.~\ref{prop-lie-derivative-varpiLie-loc} and Corollary~\ref{cor-lie-derivative-varpiLie-reductive-loc} with $\gamma$ replaced by $\mrA(X) + \gamma$, one gets the results.
\end{proof}

We shall see below in Section \ref{subsec-GRanomalies} a direct application of this local description of the Lie derivative on an Atiyah Lie algebroid.

\section{Applications}
\label{sec Applications}

Let us now make some bridges with quite recent or earlier results disseminated in the literature.

\subsection{Comparison with Crampin and Saunders' approach}
\label{sec comparison Crampin Saunder}

In \cite{CramSaun16a}, the authors have recast the notion of Cartan geometry in the language of Lie groupoids and Lie algebroids. To define \emph{infinitesimal Cartan geometry}, they introduce two Lie algebroids, which we denote 
respectively by $A\mathcal{H}$ and $A\mathcal{G}$. When these Lie algebroids come from Lie groupoids, they are realized as spaces of projectable vector fields of the fiber bundle $E \defeq \calQ \times_G G/H$ (using our notation) and the spaces of smooth sections of these Lie algebroids are isomorphic to our $\Gamma_H(\calP)$ and $\Gamma_G(\calQ)$ \cite[Sect.~3.5]{CramSaun16a}. But their \emph{infinitesimal} construction does not always reduce to this case.

They define a Cartan connection (an \emph{infinitesimal Cartan connection}) on $A\mathcal{G}$ as a splitting $\gamma : T\calM \rightarrow A\mathcal{G}$ which fulfills (in particular) the following condition \cite[Sect.~6.4]{CramSaun16a}:
\begin{equation}
\forall x\in \calM,  \gamma_x(T_x \calM)\cap(A\mathcal{H})_x = \{0\}.
\label{CrampinCondition}
\end{equation}
This condition must be compared to the well-known condition for an Ehresman connection on $\calQ$ which reduces to a Cartan connection on $\calP$ (see \cite[Sect.~A.3]{Shar97a}). This can also be compared with~\eqref{eq-ker omega cap im J = 0} established in the framework of Lie algebroids as a particular outcome of Proposition~\ref{prop Ehresmann to Cartan Lie}. 

Our approach and results focus on the following points, that are not addressed in detail in the approach developed in \cite{CramSaun16a}. Firstly, we focus our approach to Cartan connections in a “gauge field” interpretation of Atiyah Lie algebroids, as we did in \cite{FourLazzMass13a,LazzMass12a}. Next, we build a new global view of the interrelations between all the “usual” structures involved in Cartan geometry transcribed in the framework of Atiyah Lie algebroids, summarized in the diagram \eqref{eq-diagram}, in which Cartan connections are naturally embedded and characterized, see Theorem~\ref{thm eq cartan iso}. Finally, as already mentioned, our constructions rely on sections of fiber bundles without introducing, as far as possible, any geometrical (pointwise) structures. This last point can sometimes make it difficult to establish direct connections with \cite{CramSaun16a}.

On the other hand, the approach in \cite{CramSaun16a} focuses on the following points, that are not addressed by the present paper. Cartan geometries are defined at the level of Lie groupoids, even if a large part of the developments focuses on the infinitesimal version at the level of Lie algebroids \cite[Sect.~6.4]{CramSaun16a}. Moreover, their Cartan geometries are a generalization of the “ordinary” notion of Cartan geometry, as defined for  instance in \cite{Shar97a} (see \cite[Sect.~7.1~\&~7.2]{CramSaun16a} for relations with the usual approaches). Here, we definitively stay in the framework developed in \cite{Shar97a}.

Since \cite{CramSaun16a} contains a lot of developments out of the scope of this paper (some of them may be addressed in forthcoming papers), we encourage the reader to study these points directly in this book.

\subsection{A mathematical framework for gravitational anomalies}
\label{subsec-GRanomalies}

In Section~\ref{subsec-gaugeXdiff}, we have shown that a Lie derivative along $\frakX \in \Gamma_H(\calP)$ represents an infinitesimal inner gauge transformation combined with an infinitesimal diffeomorphism. Representing  in a mathematically rigorous and natural way this combined action on fields has been a challenge in theoretical physics for a while. 

For instance,  \cite{LangSchuStor84a} is devoted to the study of gravitational anomalies of the Adler--Bardeen type. The authors use the so-called BRST differential algebra of a gauge theory in which the infinitesimal local gauge parameter is turned into the Faddeev--Popov ghost field: this structure is algebraic in nature. On the other hand,  the action of the Lie algebra of $\Diff(\calM)$ (the space of smooth vector fields $\Gamma(T\calM)$ on $\calM$) is geometric in nature. In order to find a BRST treatment that encompasses both (pure) gauge transformations (inner/algebraic) and diffeomorphisms (outer/geometric), one has to find a suitable global mathematical structure in which both makes sense and can be “embedded”.

To be more precise, working in a local trivialization of a gravitation theory written in terms of the reductive Cartan geometry $\frakg = \frakh \oplus \frakp = \frakso(1,3) \oplus \bbR^{1, 3}$, the study in \cite{LangSchuStor84a} requires to define the combined action of an infinitesimal inner gauge transformation and an infinitesimal diffeomorphism on (reductive) Cartan connections. Formulas are proposed for the actions of $X \in \Gamma(\calM)$ and $v \in \Gamma_H(\calP, \frakh)$ using a background connection (we use here our notations). 

Once these actions are given, the paper then raises the following question: what is the correct geometrico-algebraic framework allowing to make sense of these relations? Let us quote here the reference~19 in \cite{Stor06a}: `\emph{These are examples where difficulties to write down WZCC signal an improper -- incomplete -- algebraic setup. The following examples also bring some water to this mill. It is tempting to require as a principle that the Wess--Zumino algebraic consistency be fulfilled in further constructions of that sort.}' In \cite{Stor93a}, Stora reiterates the fact that `\emph{the geometry involved is not widely known in physical communities}'.

The main problem is that it is not possible to make a vector field $X$ acts on forms on $\calP$: one needs a lifting, which is an extra structure on $\calP$. In fact, it is not the case that such a lifting is needed. Indeed, let us show that the present “algebraic setup”, allowing to work directly at the level of $\calP$, is the correct one to deal with this problem, and then to answer this long standing question.

\newcommand{\bfxi}{\bm{\xi}}
\newcommand{\bfomega}{\bm{\omega}}
\newcommand{\bfOmega}{\bm{\Omega}}
\newcommand{\bfmromega}{\bm{\mathring{\omega}}}
\newcommand{\bfmrR}{\bm{\mathring{R}}}
\newcommand{\bfe}{\bm{e}}
\newcommand{\bfmrnabla}{\bm{\mathring{\nabla}}}

We focus on two sets of relations given in \cite{LangSchuStor84a}, that we reproduce here with the original notations but with bold symbols to distinguish them from ours. These relations are written in a local trivialization of $\calP$ and below we refer to our notations as given in Prop.~\ref{prop infinitezimal action}. The first equations are (2a), (2b), (2c), and (2d) in \cite{LangSchuStor84a}:
\begin{equation}
\label{eq 2a-d}
\begin{aligned}
W(\bfxi) \bfomega &= - L_{\bfxi} \bfomega + \dd i_{\bfxi} \bfmromega + [\bfomega, i_{\bfxi} \bfmromega],
&
W(\bfxi) \bfe &= - L_{\bfxi} \bfe - (i_{\bfxi} \bfmromega) \bfe, 
\\
W(\bfOmega) \bfomega &= \dd \bfOmega + [\bfomega, \bfOmega],
&
W(\bfOmega) \bfe &= - \bfOmega \bfe .
\end{aligned}
\end{equation}
They express the actions of $\bfxi$ (our $X$) and $\bfOmega$ (our $\gamma$) on the Cartan connection  $\bfomega + \bfe$ (where $\bfomega$ is our $A$ and $\bfe$ is our $B$), using a fixed background connection $\bfmromega$ (our $i_\frakh \circ \mrA$) on $\calP$. 

The second set of equations, (3a), (3b), and (3c) in \cite{LangSchuStor84a}, but with the correct signs on the RHS (deduced by us from (2a), (2b), (2c), and (2d) in \cite{LangSchuStor84a}), are:
\begin{equation}
\label{eq 3a-c}
\begin{aligned}
[W(\bfOmega), W(\bfOmega')] &= - W([\bfOmega, \bfOmega']),
\\
[W(\bfxi), W(\bfOmega)] &= - W(i_{\bfxi} (\dd \bfOmega + [\bfmromega, \bfOmega])),
\\
[W(\bfxi), W(\bfxi')] &= - W([\bfxi, \bfxi']) + W(i_{\bfxi} i_{\bfxi'} \bfmrR).
\end{aligned}
\end{equation}
where $\bfmrR \defeq \dd \bfmromega + \tfrac{1}{2}[\bfmromega, \bfmromega]$ is the curvature of $\bfmromega$. These relations express the Lie brackets between the $\bfxi$ and the $\bfOmega$ induced by the previous relations. As expected from \eqref{eq 2a-d} where there is, for instance, a minus sign in front of $L_{\bfxi}$, this is an antirepresentation (hence the minus sign in the RHS) of the Lie algebra generated by $\bfxi$ and $\bfOmega$. One could have changed the sign of the action $W$ (\textit{i.e.} all the signs on the RHS of \eqref{eq 2a-d}) to get a representation: we have preferred to revise equations \eqref{eq 3a-c} for the commutators.

In order to connect \cite{LangSchuStor84a} with our results, one has to write in a local trivialization the action of $\frakX$ on $\varpiLie$ given by $L_{\frakX} \varpiLie$ as an action of $X$ and $\gamma$ on $A$ and $B$, as computed in \eqref{eq-lie-derivative-varpiLie-reductive-frakX-loc}. We have to take care of two things. Firstly, notice that $L_{\frakX}$ is a representation, while the convention in \cite{LangSchuStor84a} is to take an antirepresentation: so $W(\bfOmega)$ and $W(\bfxi)$ are to be compared to the \emph{local trivializations} of the action $-L_{\frakX}$ when $\frakX\locto \bfxi\oplus (i_{\bfxi}\bfmromega + \bfOmega)$. 
Secondly, this action has to be computed on $A$ and $B$, so that one has to consider the action on $-\varpiLie \locto A + B - i_\frakh \circ \theta_\frakh = \bfomega + \bfe - i_\frakh \circ \theta_\frakh$ and not on $\varpiLie$ (since $-A=-\bfomega$ is not a connection $1$-form), see \eqref{eq-omegaLie-betaLie-varpiLie-loc}. 

Combining these two subtleties, we consider $(-L_{\frakX}) (-\varpiLie) = L_{\frakX} \varpiLie$ in a local trivialization, which is nothing but \eqref{eq-lie-derivative-varpiLie-reductive-frakX-loc}. So, \textit{in fine}, as we can directly verify, the relations \eqref{eq 2a-d} defining the action $W$ can be exactly compared to the RHS of \eqref{eq-lie-derivative-varpiLie-reductive-frakX-loc}, where here the action of $\frakh$ on $\frakp$ is denoted by a (left) product while we used a bracket: $[i_\frakh(\gamma), B] = \gamma B$ is the action of a matrix $\gamma$ in $\frakso(1,3)$ on a vector (valued $1$-form) $B$ in $\bbR^{1, 3}$ (see details in Sect.~\ref{subsec Field theory and Lagrangian}). 

In the same way, the commutators in \eqref{eq 3a-c} can be deduced from the local trivialization of the commutator $[\frakX, \frakY]$ given in \eqref{eq local bracket} using $[-L_{\frakX}, -L_{\frakY}] = - L_{-[\frakX, \frakY]}$.

Since \eqref{eq-lie-derivative-varpiLie-reductive-frakX-loc} is the local trivialization of $L_{\frakX} \varpiLie$ (for the reductive case), one sees that equations \eqref{eq 2a-d} make sense as a unique global expression, which does not depend on the choice of a background connection. This is achieved by considering infinitesimal inner gauge transformations and infinitesimal diffeomorphisms as a single object $\frakX \in \Gamma_H(\calP)$, which mixes, by construction, algebraic and geometric properties. Notice also that one can recover in a unique way the infinitesimal diffeomorphism $X = \rho_\calP(\frakX)$ but not the infinitesimal inner gauge transformation.

A point raised in \cite{LangSchuStor84a} concerns a change of variables for $(\bfOmega, \bfxi)$ using the background connection $\bfmromega$ which produces the splitting
\begin{equation*}
\begin{tikzcd}[column sep=30pt, row sep=25pt]
\algzero
	\arrow[r]
&\Gamma_H(\calP, \frakh) 
	\arrow[r, "\iota_\calP", dashed] 
& \Gamma_H(\calP) 
	\arrow[l, bend right=25, "\bfmromega"']
	\arrow[r, "\rho_\calP", dashed] 
& \Gamma(T\calM) 
	\arrow[l, bend right=25, "\bfmrnabla"']
	\arrow[r]
& \algzero\,.
\end{tikzcd}
\end{equation*}
It is proposed in \cite{LangSchuStor84a}, equation (7), to consider $\bfOmega - i_{\bfxi}(\bfomega - \bfmromega)$ in place of $\bfOmega$. Translated in our notations, in the decomposition $\frakX = \mrnabla_X + \iota_\calP(v) \locto X \oplus (\mrA(X) + \gamma)$, this corresponds to using the connection $\nabla$ associated to $\omega$ in place of $\mrnabla$, to get $\frakX = \nabla_X + \iota_\calP(v') \locto X \oplus (A(X) + \gamma')$, since then $\gamma' = \gamma - (A(X) - \mrA(X))$. This change of variables is further discussed in \cite{Stor06a} in terms of `\emph{“Field dependent” Lie algebras}' and is described as an unsatisfactory construction: this operation `\emph{is quite strange since it mixes up parameters labeling the transformations with the field variables they act on}' \cite[p.~180]{Stor06a}. Indeed, it is clear that this trivialization of $\frakX$ as $(X, \gamma')$ depends on the (reductive) Cartan connection $\varpi = i_\frakh \circ \omega + i_\frakp \circ \beta$ (see Sect.~\ref{sec Cartan reductive}) on which $\frakX$ acts by $L_{\frakX}$. So, this change of variables is not a direction to explore further.

To conclude this section, let us emphasize once again that the Lie Algebroid framework discussed above naturally solves the long standing technical puzzles mentioned in \cite{LangSchuStor84a, Stor06a}, since it permits to define in a global way all the structures involved. This avoids the introduction of a background connection and clarify the fact that it is only required to get a local trivialization.

\subsection{Generalized Cartan connections}
\label{sec generalized cartan connections}

In \cite{FourLazzMass13a,LazzMass12a}, a notion of generalized Ehresmann connections was introduced and studied starting from a characterization of Ehresmann connections in the framework of Atiyah Lie algebroids. 

In the same spirit, from Theorem~\ref{thm eq cartan iso}, it is natural to introduce a notion of  generalized Cartan connections as isomorphisms of $C^\infty(\calM)$-modules $\varpiLie : \Gamma_H(\calP) \to \Gamma_H(\calP, \frakg)$ such that \eqref{eq-diagram-triangle-upleft} no longer commutes, \textit{i.e.} we do not assume the normalization condition $\varpiLie \circ \iota_\calP = i$, see \eqref{eq-diagram-triangle-upleft}. This notion will not be studied in detail in this paper. The mathematical study and the physical implications of such generalized Cartan connections will be addressed elsewhere.

\subsection{Field theory and Lagrangian}
\label{subsec Field theory and Lagrangian}

It is possible to recast in this Atiyah Lie algebroid language the usual Palatini action (\textit{i.e.} the Einstein-Hilbert action written in terms of tetrads). Here we only sketch the picture.

Let us consider a Cartan connection $\varpiLie$ with a reductive underlying Lie algebra $\frakg = \frakso(1,3) \oplus \mathbb{R}^{1,3}$. Thus, the Cartan connection looks like:
\begin{align*}
\varpiLie = i \circ \omegaLie + \bi \circ \betaLie.
\end{align*}
Let us consider a metric $h$ on the kernel $\Gamma_H(\calP, \frakh)$ which is $\ad$-invariant. The construction will make use of structures introduced in \cite{FourLazzMass13a}, for instance the combined integration along the kernel and the base manifold $\calM$, hereafter denoted $\int_A$.

The connection $\varpiLie$ transforms as an ordinary connection under an infinitesimal gauge transformation $v\in\Gamma_H(P;\frakh)$ as
\begin{equation*}
\delta_v \varpiLie = \hd v + [v,\varpiLie],
\end{equation*}
and its curvature transforms homogeneously
\begin{equation*}
\delta_v \bOmegaLie = [v,\OmegaLie].
\end{equation*}
Since $\frakg$ is reductive, $[\frakh,\frakp]\subset\frakp$, and both parts of $\varpiLie$ transform as
\begin{align*}
\delta_v \omegaLie = \hd v + [v,\omegaLie]
\quad
\text{and}
\quad
\delta_v \betaLie = [v,\betaLie].
\end{align*}
Recall that $v$ takes values in $\frakh$, and $\beta$ in $\frakp=\mathbb{R}^{1,3}$, thus in the matrix representation of the Lie algebra $\frakg$, the commutator $[v,\betaLie]$ reads:
\begin{equation*}
[v,\betaLie]=
\left[
\begin{pmatrix}
v & 0\\
0 & 0
\end{pmatrix}
,
\begin{pmatrix}
0 & \betaLie\\
0 & 0
\end{pmatrix}
\right]
=
\begin{pmatrix}
0 & v \betaLie\\
0 & 0
\end{pmatrix}
\end{equation*}
and it will be written as a left multiplication “$v \betaLie$”. For the same reasons, one has
\begin{align*}
[v,\betaLie]^t = - \betaLie^t v.
\end{align*}
A direct computation shows that the $\frakh$-valued $2$-form $\betaLie\wedge\betaLie^t$ transforms as:
\begin{equation*}
\delta_v (\betaLie\wedge\betaLie^t) = [v,\betaLie\wedge\betaLie^t].
\end{equation*}
Let $\OmegaLie$ be the $\frakh$-part of the curvature $\bOmegaLie$, and set the $\frakp$-part to be zero (torsionless condition). Then, the action
\begin{equation*}
S(\varpiLie)=S(\betaLie)=\int_A h(\OmegaLie,\star(\betaLie\wedge\betaLie^t))
\end{equation*}
is the Palatini action written in the language of Atiyah Lie algebroids. It is invariant due to the $\ad$-invariance of the metric $h$ together with the gauge transformations of the different parts. 

This formulation is a bit artificial, in the sense that $\betaLie\wedge\betaLie^t$ and $\OmegaLie$ have only geometric degrees of freedom as forms: so the integrand has no non trivial “inner” degree of freedom to be integrated along the kernel. However, this construction can directly be adapted to a possible generalized case, where $\varpiLie$ would be, for instance, a “generalized Cartan connection” (a mere isomorphism without assuming the normalized condition) as defined in see Sect.~\ref{sec generalized cartan connections}. In this case, $\varpiLie$ would carry inner degrees of freedom, which could make a non usual contribution to the usual Palatini action. This line of work will be investigated elsewhere.

\section{Conclusions}

In this paper we have investigated a new mathematical framework for Cartan geometry, using Atiyah Lie algebroids in their algebraic description. This led to characterize a Cartan connection as a normalized isomorphism in the commutative and exact diagram~\eqref{eq-diagram} completed by the diagram~\eqref{eq-diagram-cartan}. These two diagrams have been put in the context of the existing mathematical literature \cite{Mack87a, Shar97a, CramSaun16a} (see \textit{e.g.} Sections~\ref{sec exact comm diagram} and \ref{sec comparison Crampin Saunder}), to show to what extent they are new. Within this approach, we have moreover solved a long standing problem regarding the mathematical structures required to get the correct formulation of combined infinitesimal inner gauge transformations and infinitesimal diffeomorphisms in gravitational theories written in terms of Cartan connections \cite{LangSchuStor84a, Stor93a}. Further investigations may touch on a study of a generalized notion of Cartan connections suggested by our construction, and on the exploration of generalized Palatini-like actions.

\section*{Acknowledgments}

We warmly thank the Referee for bringing to our attention some known results on Atiyah Lie algebroids related to the construction of the diagram~\eqref{eq-diagram} in Section~\ref{sec-exact-commutative-diagram} and for pointing out to us useful references.

The project leading to this publication has received funding from Excellence Initiative of Aix-Marseille University -- A*MIDEX, in the framework of the ``Investissements d'Avenir'' French Government program managed by the French National Research Agency (ANR). J.F. was supported by the Fonds de la Recherche Scientifique --- FNRS under the grant PDR no T.0022.19.


\bibliographystyle{plainnat}
\bibliography{Biblio}
\addcontentsline{toc}{section}{References}

\end{document}